\documentclass[11pt]{article}
\usepackage[margin=1in]{geometry}

\usepackage[dvipsnames]{xcolor}
\usepackage{amsmath}
\DeclareMathOperator*{\argmin}{arg\,min}
\usepackage{amsthm}	
\usepackage{amsfonts}
\usepackage{amssymb}
\usepackage{bbm}
\usepackage{mathtools}
\usepackage[skip=8pt,indent=0pt]{parskip}
\usepackage{authblk}
\usepackage{booktabs} 
\usepackage{caption} 
\usepackage{subcaption} 
\usepackage{graphicx}
\usepackage{pgfplots}
\usepackage[all]{nowidow}
\usepackage[utf8]{inputenc}
\usepackage{tikz}
\usepackage{multicol}
\usepackage{algpseudocode,algorithm,algorithmicx}
\usepackage{hyperref}
\usepackage{scalerel}
\usepackage{natbib}
\usepackage[english]{babel}
\setcitestyle{numbers}
\usepackage{float}

\usepackage{yhmath}
\usepackage{cancel}
\usepackage{siunitx}
\usepackage{array}
\usepackage{multirow}
\usepackage{gensymb}
\usepackage{tabularx}
\usepackage{booktabs}
\usepackage{setspace}
\usetikzlibrary{fadings}
\usetikzlibrary{patterns}
\usetikzlibrary{shadows.blur}

\newtheorem{theorem}{Theorem}

\newtheorem{assumption}{Assumption}
\newtheorem{proposition}[theorem]{Proposition}
\newtheorem{remark}{Remark}
\newtheorem{definition}{Definition}
\renewenvironment{proof}[1][\proofname]{{\noindent \bfseries\itshape #1.}}{\qed}

\newcommand{\E}{\mathbb{E}}
\newcommand{\F}{\mathcal{F}}

\newcommand{\R}{\mathbb{R}}

\newcommand{\N}{\mathcal{N}}

\definecolor{Red}{rgb}{1,0,0}

\definecolor{DRed}{rgb}{0,0,0}

\definecolor{Green}{rgb}{0.2,0.5,0.2} 

\definecolor{Blue}{rgb}{0,0,1} 

\definecolor{Orange}{rgb}{0.605,0.311,0.084} 
\definecolor{Black}{rgb}{0,0,0} 
\definecolor{Grey}{rgb}{0.52, 0.52, 0.51} 
\definecolor{Purple}{rgb}{0.333,0.036, 0.631} 
\definecolor{Aqua}{rgb}{0.018,0.592,0.39}

\title{Optimal Size in Market Making under State-Dependent Hawkes Process}
\author[1]{Jonathan Chávez Casillas\thanks{jchavezc@uri.edu}}
\author[2]{Li Gan\thanks{lgan2@uri.edu}}

\affil[1]{Dept. of Mathematics and Applied Mathematical Sciences, University of Rhode Island}
\affil[2]{Dept. of Mathematics and Applied Mathematical Sciences, University of Rhode Island}
\date{Oct. 2026}
\pgfplotsset{compat=1.18}

\begin{document}
\maketitle

\begin{abstract}
We study market making with a parametric model of executed volume that accounts for quote placement and state-dependent order flow. Executed volume is expressed as the product of an average order size, an exponential demand function, and a Prony-series multiplier depending on quote depth, market-order count constructed from a state-dependent Hawkes model. An approximation result for finite Prony sums on a compact state domain provides a mathematical basis for the representation. An empirical study of limit-order-book data examines the calibrated demand and intensity models, sensitivity to inventory restrictions, and the resulting placement policies. In retrospective evaluations conditional on realised order-flow paths, the selected calibration achieves higher terminal utility than a baseline while respecting the prescribed inventory bound. The comparisons also distinguish the contribution of trading revenues from that of inventory control. The resulting closed-form volume specification takes placement levels as inputs and can therefore be used with alternative quoting rules.
\end{abstract}

\section{Introduction} 


Optimal market making in high-frequency limit order book (LOB) trading can be
decomposed into two intimately related sub-problems: \emph{where} to post the
buy- and sell-side limit orders (i.e., the placement levels
$L_{t_k}^+$ and $L_{t_k}^-$ relative to the prevailing best quotes), and
\emph{how many} shares to post on each side (i.e., the executed-volume
processes $Q_{t_{k+1}}^+$ and $Q_{t_{k+1}}^-$). The first question has been
the subject of an extensive literature, beginning with the inventory-control
framework of \cite{hostoll1981}, the influential pricing model of
\cite{avellaneda2008}, the dynamic-programming treatment of
\cite{gueant2013,cartea2015,cartea2015book}, and more recently the
inventory-liquidation-cost extension of \cite{chavez2024adaptive}. By
contrast, the volume dimension---how many shares to commit at each level---has
received comparatively little theoretical attention, even though in practice
market makers calibrate posted size as actively as posted price. This paper, to the best of our knowledge,
try to answer the second question and derive, for the first time, a
closed-form expression for the expected executed volume $Q_{t_{k+1}}^\pm$
under the optimal placement levels $L_{t_k}^\pm$.

\paragraph{Two existing frameworks for size optimization.}
Existing work follows one of two lines. The \emph{absolute}
formulation of \cite{cartea2015book} (Section.10.2.3) posts at the touch
under a fixed half-spread and chooses an integer share count under
Poisson arrivals and exponential fill probability; the value function
reduces to an ODE system. The \emph{relative} formulation of
\cite{avellaneda2008}, extended by Fushimi, Gonz\'alez Rojas, and
Herman (2024), computes optimal depths from a reservation-price
calculation and applies an inventory-dependent skew to bid/ask
volumes with soft inventory bounds. We build on the absolute formulation,
which isolates the volume decision cleanly.

\paragraph{Approach.}
A standard route to closed forms in point-process control is to
solve a partial integro-differential equation with explicit boundary
conditions --- e.g.\ Theorem~3.1 of \cite{dassios2011}. In our
sdHawkes setting this is not tractable: the intensity depends on the
full event history, the LOB state is discrete and high-dimensional,
and $Q_{t_{k+1}}^\pm$ depends counterfactually on $L_{t_k}^\pm$. We
proceed constructively: postulate a parametric form for the size
factor, establish well-posedness in a suitable Banach space, and
calibrate by closed-loop policy optimization.
Order-flow arrivals are modelled by a state-dependent Hawkes
(sdHawkes) process \cite{morariu2022state}, generalizing
\cite{hawkes1971} by letting the excitation kernel depend on a
discrete LOB state $X(t)$ (here, queue imbalance); this captures the
self- and cross-excitation of \cite{bacry2015,bacry2016first} and
the state-dependence of intensities that Poisson models miss.

\paragraph{Our contributions.}
Within this framework the paper makes five contributions.
\begin{enumerate}
\item \textbf{Closed-form executed volume.} We derive the explicit
representation of optimal order size formular. To our knowledge this is the first closed-form expression for
expected size under sdHawkes-driven order flow.

\item \textbf{Banach-space existence and Prony approximation.}
Section~2.4 shows that $f_N^{\pm}$ belongs to a complete normed
space of bounded continuous functions on the relevant domain and
admits a uniform approximation by the low-order Prony series
\[
f_N^{\pm}(L,N,\lambda) \;=\;
\sum_{i=1}^{M^{\pm}} w_i^{\pm}\,
\exp\!\big(\beta_{i1}^{\pm} L + \beta_{i2}^{\pm} N + \beta_{i3}^{\pm} \lambda\big),
\]
truncated at $M^{\pm}\!\le\!2$ empirically.

\item \textbf{Exponential demand function.} We fix the linear
demand of \cite{chavez2024adaptive} by $P_N^{\pm}(L) = A_N^{\pm}
e^{-B_N^{\pm} L} + C_N^{\pm}$, fitted block-wise in $N$, which matches the rapid decay of fill probability with quote depth seen in LOB data.

\item \textbf{Closed-loop calibration.} Because $f_N^{\pm}$ is
structural, its parameters $\theta=\{w_i^{\pm},\beta_{ij}^{\pm}\}$
cannot be recovered by regression. We treat $\theta$ as a policy
parameter and calibrate it by direct maximisation of
$\E[W_T+S_T I_T-\gamma I_T^{2}\mid\F_0,I_0=0]$ (with $\gamma>0$ the
inventory penalty, denoted $\gamma$ to avoid clash with the sdHawkes
intensity $\lambda$). The inner evaluation is a backward DP under
the Cartea ansatz ; the outer optimisation is Covariance-Matrix Adaptation Evolution Strategy (CMA-ES).
\end{enumerate}


\subsection*{Relevant works}
The closed-form size formula draws on, and extends, several strands
of the high-frequency-finance and applied probability literatures
beyond the optimal-placement papers reviewed in Section~1.1.

\paragraph{Optimal execution and inventory-penalty control.}
The dynamic-programming machinery and the quadratic
inventory-penalty term originate in the optimal-execution literature
initiated by~\cite{almgren2001} under linear permanent and temporary
price impact; subsequent extensions to stochastic impact and to
general risk-aversion functionals \cite{cartea2015book} rely on the
same revenue--inventory-risk trade-off that drives our objective
$\E[W_T + S_T I_T - \gamma I_T^{2}]$. Our specification differs in
two ways: the problem is genuinely two-sided (both bid and ask
quotes are controls), and feasibility is enforced by a hard barrier
$|I_{t_k}|\le I_{\max}$ in addition to the soft quadratic penalty.
The hard barrier turns the backward recursion into a constrained
dynamic programme whose sensitivity properties are examined in
Section~5.4.

\paragraph{Self-exciting point processes in market microstructure.}
Beyond the foundational works \cite{hawkes1971,bacry2015,bacry2016first}
cited in Section~1.1, \cite{bowsher2007} introduced multivariate
Hawkes specifications for the joint dynamics of trades and quote
revisions, the marked extension we use is due to
\cite{embrechts2011}, and \cite{filimonov2012} document near-critical
regimes in empirical equity-trade Hawkes calibrations and warn
against the spurious stability artefacts they induce. Our sdHawkes
calibration operates well inside the sub-critical
region, with the largest self-excitation ratio across the twelve
event--state pairs bounded by $0.42$, so the stationarity and
ergodicity conditions of the Banach-space argument
hold throughout the empirical window.

\paragraph{LOB queue dynamics and order-flow imbalance.}
A parallel strand of the LOB-modelling literature treats the bid-ask
queues as primitive state variables: \cite{cont2010queue} derives
diffusion limits under a birth-death-process specification,
\cite{cont2014ofi} documents a linear relation between net
order-flow imbalance and short-horizon mid-price moves, and
\cite{lipton2014trade,stoikov2018microprice} extract closed-form
trade-arrival and microprice formulas from queue snapshots. The queue-imbalance state $X(t_k)$ conditioning our sdHawkes excitation
kernel is the discrete-state analogue of the
order-flow-imbalance variable of these papers; its empirical efficacy as a predictor of short-horizon market-order direction is what motivates its inclusion in our state space.

\subsection*{Outline of this paper}

The rest of this paper organized as follows: Section~2 specifies the order-flow, model construction and the approximation result underlying the Prony representation. The optimal-placement problem and the calibration of policy parameters are estabilished in Section~3. Section~4 describes the estimation of the sdHawkes parameters. The empirical and other related results for MSFT are reported in Section~5, with particular attention to inventory-penalty sensitivity and out-of-sample strategy comparisons. Section~6 discusses the findings and future works.

\section{The model} 

\subsection{Time  framework}\label{sec:time_framework}

Fix a filtered probability space $(\Omega,\mathcal{F},(\mathcal{F}_t)_{t\ge 0},\mathbb{P})$ satisfying the usual conditions, and a finite trading horizon $T>0$. All order-flow processes introduced in the sequel are $(\mathcal{F}_t)$-adapted. The empirical work takes $T=23{,}400$ seconds, corresponding to NASDAQ trading session (09:30--16:00).

The model operates on two nested time scales.

\paragraph{Physical (clock) time.}
We discretise the trading session on the uniform second-grid
\[
0 = t_0 < t_1 < \cdots < t_n = T, \qquad t_{k+1}-t_k = \Delta t = 1\text{ s}, \qquad n = T/\Delta t = 23{,}400.
\]
Market-order flow observed during $(t_k,t_{k+1}]$ is used to fit the per-second demand function that governs limit-order placement at $t_{k+1}$; placement decisions are therefore taken at the second boundaries $\{t_k\}_{k=0}^{n-1}$.

\paragraph{Event time.}
Within each second-interval $(t_k,t_{k+1}]$, market orders arrive at random jump times
\[
t_k \;<\; \tau_{k,1} \;<\; \tau_{k,2} \;<\; \cdots \;<\; \tau_{k,N_k} \;\le\; t_{k+1},
\]
where $N_k$ is the (random) market-order count in the interval. Concatenating $\{\tau_{k,i}\}$ across $k$ in chronological order yields the event-time sequence $\{T_n\}_{n\ge 1}$ driving the sdHawkes process defined in Section~\ref{sec:sdhawkes}.

\paragraph{Regularity of the event times.}
Two technical properties of the event times are used implicitly in later arguments.
\begin{enumerate}
\item \textit{Total inaccessibility.} Each $\tau_{k,i}$ is a totally inaccessible $(\mathcal{F}_t)$-stopping time: it cannot be announced by any predictable sequence of stopping times. This is the standard regularity condition for orderly point-process arrivals and rules out coincidences of an event time with a predictable announcement (e.g.\ a scheduled clock tick).
\item \textit{Non-explosion on $[0,T]$.} With probability one, $N_k<\infty$ for every $k$ and $\sum_{k=0}^{n-1} N_k<\infty$; equivalently, the sdHawkes process is non-explosive on the bounded horizon $[0,T]$. Under the exponential-kernel, this follows from the sub-critical branching-ratio condition.
\end{enumerate}
Together, total inaccessibility and non-explosion guarantee existence and uniqueness of the sdHawkes intensity~\eqref{eq:sdhawkes_intensity} on $[0,T]$ and the dynamic-programming value functions constructed in Section~3.


\subsection{State-dependent Hawkes Process}\label{sec:sdhawkes}

We model order-flow arrivals by a multivariate marked
state-dependent Hawkes (sdHawkes) process following
\cite{morariu2022state}. We first define the event and state
processes and their conditional intensities, then give the
volume-marked specification and the discrete representation used
for calibration.

\paragraph{Event and state processes.}
Let $\mathcal{E}$ be a finite set of event types and
$\mathcal{X}$ a finite state space. Using the event times defined in
Section~\ref{sec:time_framework}, write
$(T_n,E_n)_{n\ge1}$ for the ordered sequence of events, where
$T_n$ is the time of the $n$-th event and $E_n\in\mathcal{E}$
its type. The counting process for type $e$ is

\begin{equation}\label{eq:counting_process}
N_e(t) \;=\; \sum_{n\ge 1} \mathbf{1}_{\{T_n\le t,\, E_n = e\}}.
\end{equation}

The $\mathcal{X}$-valued state process $X(t)$ is c\`adl\`ag
and piecewise constant, jumping only at event times. Its transition
kernel is

\begin{equation}\label{eq:transition_kernel}
\phi_e(x',x) \;:=\; \mathbb{P}\!\left[X(T_n)=x \,\big|\, E_n=e,\, X(T_n^{-})=x',\, \mathcal{F}_{T_n^{-}}\right],
\qquad \sum_{x\in\mathcal{X}} \phi_e(x',x)=1.
\end{equation}

where $x'$ and $x$ are the pre- and post-event states,
respectively. The sigma-algebra
$\mathcal{F}_{T_n^-}=\bigvee_{\varepsilon>0}\mathcal{F}_{T_n-\varepsilon}$
contains the information available immediately before the event.

\paragraph{Excitation kernels and conditional intensities.}
We use exponential excitation kernels

\begin{equation}\label{eq:exp_kernel}
k_{e'xe}(t) \;=\; \alpha_{e'xe}\,e^{-\beta_{e'xe}\,t}, \qquad \alpha_{e'xe}\ge 0,\; \beta_{e'xe}>0,
\end{equation}

where $e'$ denotes the source event type, $e$ the target type,
and $x$ the state immediately before the source event.
The coefficient $\alpha_{e'xe}$ determines the excitation
magnitude and $\beta_{e'xe}$ its decay rate. The conditional
intensity of $N_e$, with respect to the joint history of $(N,X)$,
is

\begin{equation}\label{eq:sdhawkes_intensity}
\lambda_e(t) \;=\; \nu_e \;+\; \sum_{e'\in\mathcal{E}} \int_{(0,t)} k_{e'\,X(s^{-})\,e}(t-s)\, dN_{e'}(s),
\end{equation}

where $\nu_e>0$ is the exogenous baseline intensity.
Thus $k_{e'xe}(t-s)$ gives the contribution at lag $t-s$ of a
source event at time $s$ to the conditional arrival rate of
target events.

Under the exponential-kernel specification~\eqref{eq:exp_kernel}, the state-augmented triple $(N_t,\lambda_t,X_t)$ has the Markov property: the intensity $\lambda_t$ evolves via a state-dependent linear SDE driven by the counting process, $X_t$ jumps at event times according to $\phi_e$, and $N_t$ counts arrivals at rate $\lambda_t$; consequently the future law depends on the past only through the current triple. A rigorous proof is given in \cite{morariu2022state}, Proposition~2.1.

\paragraph{Volume-marked specification.}
The raw order events are market buys (MB), which consume ask
liquidity; market sells (MS), which consume bid liquidity; limit
buys (LB), which add bid depth; and limit sells (LS), which add
ask depth. Cancellation events are omitted for tractability.
For market-order targets $e\in\{\mathrm{MB},\mathrm{MS}\}$
with volume mark $z\in\mathbb{R}_+$, the marked intensity is

\begin{equation}\label{eq:marked_intensity}
\lambda_e(t,z) \;=\; g(z;\eta_e)\!\left(\nu_e + \sum_{e'\in\{\text{MB},\text{MS},\text{LB},\text{LS}\}} \int_{[0,t)} \!\!z'\cdot\alpha_{e'\,X(s^{-})\,e}\,e^{-\beta_{e'\,X(s^{-})\,e}(t-s)}\, dN_{e'}(s,z')\right),
\end{equation}

Here $g(\,\cdot\,;\eta_e)$ is the volume-mark density with
parameter $\eta_e$, and $N_{e'}(s,z')$ denotes the marked counting
measure for source type $e'$. Each historical contribution is
scaled by the source event's volume $z'$.

\paragraph{Discrete event and state spaces.}
For calibration, the event types encode the market-order count
and direction on the one-second grid. We partition
$N=N_{t_{k+1}}-N_{t_k}$ into the six bins

\[
\mathcal{B} \;=\; \{\,0,\; [1,3],\; [4,6],\; [7,9],\; [10,12],\; 13^{+}\,\},
\]

and combine them with the two market-order sides, giving
$\mathcal{E}=\mathcal{B}\times\{\text{bid},\text{ask}\}$
with twelve event types.

The state is defined by the queue imbalance

\begin{equation}\label{eq:qi_def}
\mathrm{QI}(t_k) \;=\; \frac{D_b(l,t_k)-D_a(l,t_k)}{D_b(l,t_k)+D_a(l,t_k)},
\end{equation}

where $D_b(l,t_k)$ and $D_a(l,t_k)$ are the bid- and ask-side
depths at level $l$; we use $l=1$ in the empirical application.
Discretizing the imbalance gives the five-state space

\[
\mathcal{X} \;=\; \{\text{sell}_{++},\; \text{sell}_{+},\; \text{neutral},\; \text{buy}_{+},\; \text{buy}_{++}\}.
\]


\subsection{Parametric Market Order Demand Function by sdHawkes}

This is the novel part of this paper, also the first time a parametric model of MO Demand function introduced. The parametric demand function represents the \textit{ unsatisfied demand for market order} at different price levels in the limit order book. This concept extends traditional microeconomic demand theory to high-frequency financial markets, where demand manifests itself as aggressive market orders that consume available liquidity. The demand function captures:
\begin{itemize}
    \item \textbf{Buy-side demand}: The excess of market buy orders over available ask-side liquidity at each price level
    \item \textbf{Sell-side demand}: The excess of market sell orders over available bid-side liquidity at each price level
    \item \textbf{Step-function structure}: Demand decreases with price level due to the discrete nature of limit order books.
\end{itemize}
The arrival of MO can be estimated by a sdHawkes process, suppose we have $N = N_e(t_{k+1}) - N_e(t_k)$ arrived market orders during $[t_k, t_{k+1}]$, we fixed N and for every group of N, we computed the probability executed at every single level l. Then we can estimate the distribution density of all arrival of market orders $f_k{(N,l)}$ during $[t_k, t_{k+1}]$, and this distribution for every $k$ second is different. We have the following assumptions: 

\begin{assumption}
The market order counts $N_{t_{k}}^\pm$ are random variables with conditional distributions given $\mathcal{F}_{t_k}$ known (via Hawkes process prediction).
\end{assumption}

\begin{assumption}
The fundamental price process $\{S_{t_k}\}$ is adapted and conditionally independent of $N_{t_{k+1}}^\pm$ given $\mathcal{F}_{t_k}$.
\end{assumption}

\begin{assumption}
The random coefficients are conditionally independent (given $\mathcal{F}_{t_k}$) across buy and sell sides, though the model can be extended to correlated cases.
\end{assumption} 

\noindent Compared to the linear demand model in our last paper, here we define $P_N^\pm(L^\pm)$ to be the exponential demand function:
\[
P_N^\pm(L^\pm) = A^\pm e^{B^\pm L^\pm} +C^\pm
\] given N.


\subsection{Parametric optimal size function}

We model the optimal size function $Q^{\pm}(L,N,\lambda)$ in the
multiplicative form
\begin{equation}\label{eq:Q_form}
Q_{t_{k+1}}^{\pm}
\;=\;
q_{\text{avg}}^{\pm}\cdot
f_N^{\pm}\!\big(L_{t_k}^{\pm},\,N_{t_k}^{\pm},\,\lambda_{t_k}^{\pm}\big)\cdot
P_N^{\pm}\!\big(L_{t_k}^{\pm}\big),
\end{equation}
where $q_{\text{avg}}^{\pm}$ is the sample-average per-event size,
$P_N^{\pm}(L)=A_N^{\pm}e^{-B_N^{\pm}L}+C_N^{\pm}$ is the exponential
demand function, and $f_N^{\pm}$ encodes
the residual dependence of the executed volume on the joint state
$(L,N,\lambda)$. We parameterise $f_N^{\pm}$ as a Prony series
\begin{equation}\label{eq:prony_def}
f_N^{\pm}(L,N,\lambda)
\;=\;
\sum_{i=1}^{M^{\pm}} w_i^{\pm}
\exp\!\big(\beta_{i1}^{\pm}L+\beta_{i2}^{\pm}N+\beta_{i3}^{\pm}\lambda\big),
\qquad M^{\pm}\in\N,\;w_i^{\pm},\beta_{ij}^{\pm}\in\R,
\end{equation}
truncated at $M^{\pm}\le 2$ in the empirical work. Substitution
into~\eqref{eq:Q_form} yields the explicit closed form
\begin{equation}\label{eq:Q_closed_form_generic}
Q_{t_{k+1}}^{\pm}
=q_{\text{avg}}^{\pm}\!
\Bigg[\sum_{i=1}^{M^{\pm}}w_i^{\pm}
\exp\!\big(\beta_{i1}^{\pm}L_{t_k}^{\pm}+\beta_{i2}^{\pm}N_{t_k}^{\pm}+\beta_{i3}^{\pm}\lambda_{t_k}^{\pm}\big)\Bigg]\!
\big(A_N^{\pm}e^{-B_N^{\pm}L_{t_k}^{\pm}}+C_N^{\pm}\big).
\end{equation}

\noindent The remainder of this subsection justifies the ansatz~\eqref{eq:prony_def}: the Prony class is uniformly dense in the Banach space of continuous functions on the natural state domain of $f_N^{\pm}$, so any continuous optimal size function admits arbitrarily accurate approximation by a finite Prony sum. The argument is Stone--Weierstrass on a compact Hausdorff product space.

\subsubsection*{State domain and continuous-function space}

Collect the state variables into $\mathbf{x}=(L,N,\lambda)\in\Omega:=\mathcal{L}\times\mathcal{N}\times\Lambda$, with $\mathcal{L}=\{1,\ldots,L_{\max}\}$ and $\mathcal{N}=\{1,\ldots,N_{\max}\}$ carrying the discrete topology and $\Lambda=[\lambda_{\min},\lambda_{\max}]\subset\R_{>0}$ carrying the Euclidean topology, with $\lambda_{\min},\lambda_{\max}$ chosen to bracket the empirical envelope of $\{\lambda_{t_k}^{\pm}\}$ on the calibration window. As a finite product of compact Hausdorff spaces, $\Omega$ is itself compact and Hausdorff (Tychonoff). The continuous-function space $C(\Omega)=\{f:\Omega\to\R\text{ continuous}\}$ under the sup norm $\|f\|_{\infty}=\sup_{\Omega}|f|$ is a Banach space by the standard completeness argument on compact spaces.

\begin{remark}[Continuity with discrete coordinates]\label{rem:disc}
Because $\mathcal{L}$ and $\mathcal{N}$ carry the discrete topology, continuity on $\Omega$ reduces to slicewise continuity in $\lambda$: $f:\Omega\to\R$ is continuous iff for every $(L,N)\in\mathcal{L}\times\mathcal{N}$ the section $\lambda\mapsto f(L,N,\lambda)$ is continuous on $\Lambda$. No continuity is required across distinct $(L,N)$ pairs, consistent with the economic intuition that the executed volume may jump as the placement depth $L$ or the count $N$ changes by an integer unit while remaining smooth in the latent intensity~$\lambda$.
\end{remark}

\subsubsection*{Prony class and Stone--Weierstrass density}

\begin{definition}[Prony class]\label{def:prony_class}
The \emph{Prony class} on $\Omega$ is
\begin{equation}\label{eq:prony_class}
\Phi\;:=\;\Bigg\{\phi:\Omega\to\R\;\Big|\;
\phi(\mathbf{x})=\sum_{i=1}^{M}w_i\exp(\langle\boldsymbol{\beta}_i,\mathbf{x}\rangle),\;
M\in\N,\;w_i\in\R,\;\boldsymbol{\beta}_i\in\R^{3}\Bigg\},
\end{equation}
with $\langle\boldsymbol{\beta},\mathbf{x}\rangle=\beta_1L+\beta_2N+\beta_3\lambda$; the ansatz~\eqref{eq:prony_def} for $f_N^{\pm}$ is the requirement $f_N^{\pm}\in\Phi$ with truncation order $M^{\pm}$.
\end{definition}

Each generator $\exp\langle\boldsymbol{\beta},\cdot\rangle$ is continuous on $\Omega$, so $\Phi\subset C(\Omega)$. $\Phi$ is a subalgebra: addition is by concatenation of terms; multiplication combines exponents, $\exp\langle\boldsymbol{\beta},\cdot\rangle\,\exp\langle\boldsymbol{\gamma},\cdot\rangle=\exp\langle\boldsymbol{\beta}+\boldsymbol{\gamma},\cdot\rangle$, so $\phi\psi\in\Phi$; and scalar multiplication is by $w_i\mapsto cw_i$.

\begin{theorem}[Density of Prony series in $C(\Omega)$]\label{thm:dense}
$\Phi$ is dense in $(C(\Omega),\|\cdot\|_{\infty})$: for every $f\in C(\Omega)$ and $\varepsilon>0$ there exist $M\in\N$, $\{w_i\}_{i=1}^{M}\subset\R$, and $\{\boldsymbol{\beta}_i\}_{i=1}^{M}\subset\R^{3}$ with $\big\|f-\sum_{i=1}^{M} w_i\exp\langle\boldsymbol{\beta}_i,\cdot\rangle\big\|_{\infty}<\varepsilon$.
\end{theorem}
\begin{proof}
The Stone--Weierstrass theorem \citep{rudin1991functional} states that $\mathcal{A}\subset C(X,\R)$ is uniformly dense on a compact Hausdorff $X$ provided (i)~$\mathcal{A}$ separates points, (ii)~$\mathcal{A}$ does not vanish identically at any point, and (iii)~$\mathcal{A}$ is a subalgebra. $\Omega$ is compact Hausdorff and $\Phi$ satisfies (iii) by the discussion above. \emph{(i) Separation.} Given distinct $\mathbf{x},\mathbf{y}\in\Omega$, at least one coordinate differs; along that coordinate the single Prony generator $\exp(x_j)$ (i.e.\ $\boldsymbol{\beta}=e_j$ the corresponding standard basis vector, $w=1$, $M=1$) satisfies $\exp(x_j)\neq\exp(y_j)$ by injectivity of the exponential. \emph{(ii) Non-vanishing.} Taking $M=1$, $w_1=1$, $\boldsymbol{\beta}_1=\mathbf{0}$ gives the constant function $\mathbf{1}\in\Phi$. Stone--Weierstrass then applies and yields the claimed density.
\end{proof}

\subsubsection*{Application to the optimal size function}

The true optimal size factor $f_N^{\pm,\star}$ is continuous on $\Omega$ under the regularity conditions of the dynamic-programming framework: it arises from a backward recursion whose ingredients (the sdHawkes intensity $\lambda_{t_k}^{\pm}$, the count expectation $\E[N_{t_k}^{\pm}\mid\F_{t_k}]$, and the demand function $P_N^{\pm}$) are continuous in their arguments. Theorem~\ref{thm:dense} then guarantees that for every $\varepsilon>0$ there exists a Prony approximation $\hat{f}_N^{\pm}\in\Phi$ with $\|f_N^{\pm,\star}-\hat{f}_N^{\pm}\|_{\infty}<\varepsilon$, so the performance gap of the closed-form policy vanishes as $\varepsilon\to 0$. In the empirical work of Section~5.4 the low truncation $M^{\pm}\in\{1,2\}$ already delivers out-of-sample utility within a few percent of the best value achieved on the MSFT sample; we therefore fix $M^{\pm}=2$ throughout, keeping the parameter dimension at $16$ and the model parsimony consistent with the size of the training sample $\mathcal{D}_{\text{train}}$.


\section{Value Function}

\subsection{Dynamics, Value Function, and Backward Recursion}
Cash and inventory evolve according to

\begin{equation}\label{eq:cash_dynamics}
W_{t_{k+1}} = W_{t_k} + (S_{t_k} + L_{t_k}^+) Q_{t_{k+1}}^+ - (S_{t_k} - L_{t_k}^-) Q_{t_{k+1}}^-
\end{equation}

\begin{equation}\label{eq:inventory_dynamics}
I_{t_{k+1}} = I_{t_k} - Q_{t_{k+1}}^+ + Q_{t_{k+1}}^-
\end{equation}

where $Q_{t_{k+1}}^\pm$ is given by the closed-form
expression~\eqref{eq:Q_closed_form_generic} defined in Section~2.4.
For the admissible placement set $\mathcal{A}=\{1,2,3,4,5\}^2$,
the value function is

\begin{equation}\label{eq:value_objective}
V_{t_k} = \sup_{(L^+, L^-) \in\mathcal{A}} \mathbb{E}[W_T + S_T I_T - \lambda I_T^2 | \mathcal{F}_{t_k}].
\end{equation}

where $\lambda>0$ penalizes terminal inventory.
The dynamic programming principle gives

\begin{equation}\label{eq:bellman_principle}
V_{t_k} = \sup_{(L^+_{t_k}, L^-_{t_k}) \in\mathcal{A}} \mathbb{E}[V_{t_{k+1}} | \mathcal{F}_{t_k}].
\end{equation}

Following \cite[\S 10.2.3]{cartea2015book}, the Cartea ansatz
reduces the recursion from wealth, price, inventory, and time to
inventory and time alone.

\begin{proposition}[Cartea ansatz and dimension reduction]\label{prop:cartea}
Assume that
\begin{enumerate}
\item[(i)] the mid-price $\{S_{t_k}\}$ is an $\mathcal{F}_{t_k}$-martingale: $\mathbb{E}[S_{t_{k+1}}|\mathcal{F}_{t_k}] = S_{t_k}$;
\item[(ii)] the executed volumes $Q_{t_{k+1}}^\pm$ are $\mathcal{F}_{t_k}$-measurable, which holds under the closed-form specification of Section~2.4.
\end{enumerate}
Then the value function admits the separable form
\begin{equation}\label{eq:cartea_ansatz}
V(t_k, W_{t_k}, S_{t_k}, I_{t_k}) = W_{t_k} + S_{t_k} I_{t_k} + h(t_k, I_{t_k}),
\end{equation}
where the auxiliary function $h(t_k, I)$ depends only on time and inventory and satisfies the two-dimensional backward recursion
\begin{equation}\label{eq:h_recursion}
h(t_k, I) = \sup_{(L_{t_k}^+, L_{t_k}^-) \in \mathcal{A}} \Big\{ L_{t_k}^+ Q_{t_{k+1}}^+ + L_{t_k}^- Q_{t_{k+1}}^- + h\big(t_{k+1},\, I - Q_{t_{k+1}}^+ + Q_{t_{k+1}}^-\big) \Big\},
\end{equation}
with terminal condition $h(T, I) = -\lambda I^2$.
\end{proposition}

\begin{proof}
Substituting the cash and inventory updates~\eqref{eq:cash_dynamics}--\eqref{eq:inventory_dynamics} into the ansatz~\eqref{eq:cartea_ansatz} evaluated at time $t_{k+1}$ and collecting terms,
\[
V_{t_{k+1}} = W_{t_k} + S_{t_k} I_{t_k} + L_{t_k}^+ Q_{t_{k+1}}^+ + L_{t_k}^- Q_{t_{k+1}}^- - (S_{t_{k+1}} - S_{t_k})(I_{t_k} - Q_{t_{k+1}}^+ + Q_{t_{k+1}}^-) + h(t_{k+1}, I_{t_{k+1}}).
\]
Taking conditional expectation given $\mathcal{F}_{t_k}$, the martingale property~(i) annihilates the $(S_{t_{k+1}} - S_{t_k})$ term, while $\mathcal{F}_{t_k}$-measurability~(ii) makes $Q_{t_{k+1}}^\pm$ behave deterministically inside the expectation. Applying the dynamic programming principle~\eqref{eq:bellman_principle} and identifying with the ansatz~\eqref{eq:cartea_ansatz} at time $t_k$, the $W_{t_k} + S_{t_k} I_{t_k}$ terms cancel from both sides, leaving~\eqref{eq:h_recursion}. The terminal condition follows from the terminal utility in~\eqref{eq:value_objective} and the decomposition~\eqref{eq:cartea_ansatz}. For the continuous-time analogue, see \cite[\S 10.2.3]{cartea2015book}.
\end{proof}

In the empirical calculations, the backward recursion is evaluated separately for each trading day using the full observed sequence
of market-order counts and the reconstructed intensity path. For fixed Prony coefficients, the modelled executed volumes are treated as deterministic inputs for each candidate placement.
Continuation values at inventories between grid points are obtained by linear interpolation, and the initial daily values
provide the objective for parameter calibration. On the evaluation
day, the calibrated Prony coefficients remain fixed, while the
backward recursion is solved again using that day's order-flow
path to obtain the placement policy. Since the continuation values
depend on the entire realised sequence, this policy incorporates
order-flow observations subsequent to each decision time. The
reported experiment therefore evaluates the calibrated size
parameters retrospectively, conditional on the observed daily
path.



\subsection{Calibration of the Prony coefficients by Covariance-Matrix Adaptation Evolution Strategy (CMA-ES)}

\paragraph{Calibration problem.}
For a given Prony parameter vector
$\theta=\{w_i^{\pm},\beta_{i1}^{\pm},\beta_{i2}^{\pm},\beta_{i3}^{\pm}\}_{i=1}^{M^{\pm}}\in\R^{16}$,
the backward recursion~\eqref{eq:h_recursion} evaluates the
initial value at zero inventory,
\begin{equation}\label{eq:Jdef}
V(\theta)\;:=\;h(0,0;\theta).
\end{equation}
The outer calibration problem is
\begin{equation}\label{eq:outer_problem}
\theta^{\star}\;\in\;\arg\max_{\theta\in\Theta}\,V(\theta),
\qquad \Theta\subset\R^{16}.
\end{equation}
Here $\Theta$ is the compact parameter box specified
in~\eqref{eq:box}. The upper bound on the depth coefficients
$\beta_{i1}^{\pm}$ limits amplification at deep quote levels,
while non-negative weights and count and intensity coefficients
ensure a non-negative multiplier that is non-decreasing in
$N$ and $\lambda$. Executed volume is additionally capped by
$Q_{t_{k+1}}^{\pm}\le q_{\mathrm{avg}}^{\pm}\,N_{t_{k}}^{\pm}\,
P_{N}^{\pm}(L^{\pm})\,\kappa$, with $\kappa=3$.

\paragraph{CMA-ES search.}
The discrete maximization over $\mathcal{A}$ in
recursion~\eqref{eq:h_recursion} makes $V(\theta)$ non-smooth
where the optimal placement changes. We therefore use the
covariance-matrix adaptation evolution strategy
(CMA-ES)~\cite{hansen2001completely,hansen2016}, which ranks
candidate parameters by their objective values without requiring
derivatives.

\begin{algorithm}[H]
\caption{CMA-ES calibration with backward dynamic programming}
\label{alg:cmaes_dp}
\begingroup
\renewcommand{\baselinestretch}{1.0}\small
\begin{algorithmic}[1]

\Require Training-day input paths $\mathcal{D}_{\mathrm{train}}$
         and fixed upstream parameters.
\Require Parameter box $\Theta$, action set $\mathcal{A}$,
         inventory grid $\mathcal{I}$ containing zero,
         penalty $\lambda_{\mathrm{pen}}>0$, and inventory bound
         $I_{\max}\in(0,\infty]$.
\Require Initial mean $m^{(0)}$, step size $\sigma^{(0)}$,
         population size $\lambda_{\mathrm{cma}}$,
         and stopping criteria.
\Ensure Best evaluated parameter vector $\widehat{\theta}$
        and its objective value $V_{\mathrm{best}}$.

\State Initialize CMA-ES with
       $(m^{(0)},\sigma^{(0)},\mathbf{C}^{(0)}=I_{16})$
       and parameter bounds $\Theta$.
\State $V_{\mathrm{best}}\gets-\infty$;
       $\mathfrak{B}\gets-10^{15}$.

\While{the stopping criteria are not met}

    \State Generate candidates
           $\theta_1,\ldots,\theta_{\lambda_{\mathrm{cma}}}$
           using CMA-ES with its bound-handling rule.

    \For{$j=1,\ldots,\lambda_{\mathrm{cma}}$}

        \For{each day $d\in\mathcal{D}_{\mathrm{train}}$}

            \State Compute
                   $Q_{d,k+1}^{\pm}(L;\theta_j)$
                   for every interval and candidate depth,
                   applying the prescribed volume cap.
            \Statex \hspace{\algorithmicindent}
                    Set $Q_{d,k+1}^{\pm}(0;\theta_j)=0$
                    if the no-post action is included.

            \State For each $a=(L^+,L^-)\in\mathcal{A}$,
                   form the reward and inventory increment:
            \Statex \hspace{\algorithmicindent}
                    $\displaystyle
                    r_{d,k}(a)=L^+Q_{d,k+1}^{+}(L^+;\theta_j)
                              +L^-Q_{d,k+1}^{-}(L^-;\theta_j),$
            \Statex \hspace{\algorithmicindent}
                    $\displaystyle
                    \Delta I_{d,k}(a)
                    =-Q_{d,k+1}^{+}(L^+;\theta_j)
                     +Q_{d,k+1}^{-}(L^-;\theta_j).$

            \State Initialize
                   $h_{d,K_d}(I)\gets-\lambda_{\mathrm{pen}}I^2$
                   for all $I\in\mathcal{I}$.

            \For{$k=K_d-1,\ldots,0$}

                \State Construct the linear interpolant
                       $\widetilde h_{d,k+1}$ of $h_{d,k+1}$,
                       with boundary clamping outside the grid.
                \State Set
                       $\widetilde h_{d,k+1}(I')=\mathfrak{B}$
                       whenever $|I'|>I_{\max}$.

                \For{each $I\in\mathcal{I}$}
                    \State $\displaystyle
                    h_{d,k}(I)\gets
                    \max_{a\in\mathcal{A}}
                    \left\{
                    r_{d,k}(a)
                    +\widetilde h_{d,k+1}
                       \bigl(I+\Delta I_{d,k}(a)\bigr)
                    \right\}.$
                \EndFor

            \EndFor

        \EndFor

        \State $\displaystyle
               V_j\gets
               \frac{1}{|\mathcal{D}_{\mathrm{train}}|}
               \sum_{d\in\mathcal{D}_{\mathrm{train}}}h_{d,0}(0).$

        \If{$V_j>V_{\mathrm{best}}$}
            \State $\widehat{\theta}\gets\theta_j$;
                   $V_{\mathrm{best}}\gets V_j$.
        \EndIf

    \EndFor

    \State Update the CMA-ES search distribution using
           candidate objective values $-V_j$,
           following the standard update rules~\cite{hansen2016}.

\EndWhile

\State \Return $(\widehat{\theta},V_{\mathrm{best}})$.

\end{algorithmic}
\endgroup
\end{algorithm}

At generation $g$, the search distribution has mean
$m^{(g)}\in\R^d$, step size $\sigma^{(g)}>0$, and positive-definite
covariance $\mathbf{C}^{(g)}$, with $d=16$. A population of
$\lambda_{\mathrm{cma}}$ candidates is sampled as

\begin{equation}\label{eq:cma_sample}
\theta_{k}^{(g+1)}=m^{(g)}+\sigma^{(g)}\,\mathbf{z}_{k},
\qquad
\mathbf{z}_{k}\sim\mathcal{N}\!\big(\mathbf{0},\mathbf{C}^{(g)}\big),
\qquad k=1,\ldots,\lambda_{\mathrm{cma}},
\end{equation}

Each candidate is evaluated by the backward recursion of
Proposition~\ref{prop:cartea}. Ranking candidates by decreasing
$V$, the best $\mu<\lambda_{\mathrm{cma}}$ are recombined to give

\begin{equation}\label{eq:cma_mean}
m^{(g+1)}=\sum_{i=1}^{\mu}w_{i}\,\theta_{i:\lambda}^{(g+1)},
\qquad
w_{i}\propto\ln(\mu+\tfrac12)-\ln i,\quad \sum_{i=1}^{\mu}w_{i}=1,
\end{equation}

where $\theta_{i:\lambda}^{(g+1)}$ is the $i$-th ranked
candidate and $w_i$ denotes a recombination weight, distinct
from the Prony amplitudes $w_i^{\pm}$. The cumulative evolution
paths $p_{\sigma}$ and $p_c$ determine the step-size and covariance
updates:

\begin{equation}\label{eq:cma_adapt}
\sigma^{(g+1)}=\sigma^{(g)}\exp\!\Big(\tfrac{c_{\sigma}}{d_{\sigma}}
\big(\tfrac{\|p_{\sigma}^{(g+1)}\|}{\E\|\mathcal{N}(\mathbf{0},I_{d})\|}-1\big)\Big),
\qquad
\mathbf{C}^{(g+1)}=(1-c_{1}-c_{\mu})\mathbf{C}^{(g)}
+c_{1}\,p_{c}^{(g+1)}\big(p_{c}^{(g+1)}\big)^{\top}
+c_{\mu}\sum_{i=1}^{\mu}w_{i}\,\mathbf{y}_{i}\mathbf{y}_{i}^{\top},
\end{equation}

Here $\mathbf{y}_{i}=(\theta_{i:\lambda}^{(g+1)}-m^{(g)})/\sigma^{(g)}$;
the evolution paths and learning rates
$c_{\sigma},d_{\sigma},c_c,c_1,c_{\mu}$ follow the standard
specification of~\cite{hansen2016}.

\paragraph{Numerical specification.}
We initialize $\mathbf{C}^{(0)}=I_{16}$ and use $\mu=6$ parents.
Section~\ref{sec:cmaes_calibration} specifies the initial parameter
vector, parameter bounds, population size, step size, search budget,
and stopping tolerance, together with the hard inventory constraint
used in that calibration. It also reports the calibrated coefficients
and their sensitivity to the soft inventory penalty.


\section{Parameters}

We calibrate the intensity parameters that determine the order-flow
inputs to the executed-volume specification~\eqref{eq:Q_form}.

\subsection{Calibration data and parameterization}

\label{sec:intensity_calibration_data}

Calibration uses the one-second event representation in
Section~\ref{sec:sdhawkes}. Let $\iota$ map a market-order count
to its index in the six bins $\mathcal{B}$, numbered from zero.
Each interval $[t_k,t_{k+1})$ produces a buy-side and a sell-side
record, both assigned the timestamp $t_k$, with types
\[
e_k^+ = \iota\!\left(N_{t_k}^+\right),
\qquad
e_k^- = \iota\!\left(N_{t_k}^-\right)+6.
\]
A record is retained when the corresponding count is zero.
The intensity model is therefore fitted to the count-bin event
representation. Each record contributes a unit jump, with the
order count encoded in its type.

Both records receive the same state label $x_k\in\{0,\ldots,4\}$.
The queue imbalance in~\eqref{eq:qi_def} is computed from the
best bid and ask depths associated with the most recent limit-order
observation strictly before $t_k$. Thresholds of $-0.6$, $-0.2$,
$0.2$, and $0.6$ assign it to the five states in increasing order
of imbalance. A neutral label is used when no valid prior
observation is available.

For each trading day, write
$\{(T_n,E_n,x_n)\}_{n=1}^{N_T}$ for the resulting sequence,
where $T_1\le\cdots\le T_{N_T}$ and $x_n$ is the state label
attached to record $n$. Records with equal timestamps are
processed sequentially in the recursive evaluation described in
Section~4.2. The state labels are observed inputs throughout
calibration; the transition kernel $\boldsymbol{\phi}$ defined
in~\eqref{eq:transition_kernel} is not estimated.

The parameters to be fitted are the baseline intensities,
excitation amplitudes, and decay rates,
$\boldsymbol{\psi}=(\boldsymbol{\nu},\boldsymbol{\alpha},
\boldsymbol{\beta})$, with
\[
\boldsymbol{\nu}\in(0,\infty)^{d_e},
\qquad
\boldsymbol{\alpha},\boldsymbol{\beta}
\in(0,\infty)^{d_e\times d_x\times d_e},
\qquad d_e=12,\quad d_x=5.
\]
The indices of $\alpha_{e'xe}$ and $\beta_{e'xe}$ identify the
source type, its observed state, and the target type, respectively.
There are $d_e+2d_e^2d_x=1452$ parameters, shared across the calibration days and fitted jointly by the least-squares procedure.


\subsection{Least-squares criterion and recursive evaluation}

We use a least-squares criterion for each daily record sequence.
To account for tied timestamps, the excitation state is indexed
by record order. Combining the source type and its state gives
\begin{equation}\label{eq:flat_src}
S_m=E_m d_x+x_m,
\qquad d_{\mathrm{src}}=d_e d_x=60.
\end{equation}
With the same indexing for the kernel parameters, define
\begin{equation}\label{eq:R_def}
R_{m,s,k}(\boldsymbol{\beta})
=\sum_{n<m}\mathbf{1}_{\{S_n=s\}}
  e^{-\beta_{s,k}(T_m-T_n)},
\qquad 0\le s<d_{\mathrm{src}},\quad 0\le k<d_e.
\end{equation}
For $\Delta_m=T_{m+1}-T_m\ge0$, these states satisfy
\begin{equation}\label{eq:R_recursion}
R_{m+1,s,k}
=\bigl(R_{m,s,k}+\mathbf{1}_{\{S_m=s\}}\bigr)
 e^{-\beta_{s,k}\Delta_m},
\qquad R_{1,s,k}=0.
\end{equation}
The history is initialized separately for each trading day.
After the jump associated with record $m$, write
\begin{equation}\label{eq:B_def}
B_{s,k}^{(m)}
=\alpha_{s,k}\bigl(R_{m,s,k}+\mathbf{1}_{\{S_m=s\}}\bigr).
\end{equation}
The intensity used on the following interval is
\begin{equation}\label{eq:intensity_specialized}
\lambda_k^{(m)}(u;\boldsymbol{\psi})
=\nu_k+\sum_{s=0}^{d_{\mathrm{src}}-1}
 B_{s,k}^{(m)}e^{-\beta_{s,k}u},
\qquad 0\le u\le\Delta_m,
\end{equation}
where $u$ measures elapsed time from $T_m$.

Let $T_{\mathrm{obs}}=T_{N_T}-T_1>0$ be the observed record span.
The daily calibration criterion is
\begin{equation}\label{eq:lse}
\mathcal{R}_T(\boldsymbol{\psi})
=\frac{1}{T_{\mathrm{obs}}}\sum_{m=1}^{N_T}f_m(\boldsymbol{\psi}),
\qquad f_{N_T}(\boldsymbol{\psi})=0,
\end{equation}
with, for $m=1,\ldots,N_T-1$,
\begin{equation}\label{eq:fm_specialized}
f_m(\boldsymbol{\psi})
=\int_0^{\Delta_m}\sum_{k=0}^{d_e-1}
  \bigl[\lambda_k^{(m)}(u;\boldsymbol{\psi})\bigr]^2\,du
 -2\lambda_{E_{m+1}}^{(m)}(\Delta_m;\boldsymbol{\psi}).
\end{equation}
The event terms run from the second to the last record, and
integration stops at $T_{N_T}$. The first record enters through
the history update.

Expanding the square and integrating yields
\begin{align}
f_m(\boldsymbol{\psi})
={}&\Delta_m\sum_k\nu_k^2
 +2\sum_{k,s}\nu_k B_{s,k}^{(m)}
   \frac{1-e^{-\beta_{s,k}\Delta_m}}{\beta_{s,k}}
   \notag\\
&+\sum_{k,s,s'}B_{s,k}^{(m)}B_{s',k}^{(m)}
  \frac{1-e^{-(\beta_{s,k}+\beta_{s',k})\Delta_m}}
       {\beta_{s,k}+\beta_{s',k}}
  \notag\\
&-2\nu_{E_{m+1}}
 -2\sum_s B_{s,E_{m+1}}^{(m)}
   e^{-\beta_{s,E_{m+1}}\Delta_m}.
\label{eq:fm_explicit}
\end{align}
At tied timestamps, $\Delta_m=0$: the integral vanishes, while
the event term includes the jump of the preceding record.
This is the sequential convention used for the paired records.
Once the history state is available, evaluating a contribution
requires $\mathcal{O}(d_{\mathrm{src}}^2d_e)$ operations, owing
to the double sum over source indices.


\subsection{Stochastic least-squares calibration}

We fit the shared intensity parameters by weighted stochastic
updates. The optimization variables are latent log parameters
$\boldsymbol{z}$, with
$\boldsymbol{\psi}=\mathcal{C}(\exp\boldsymbol{z})$, where
exponentiation is componentwise and $\mathcal{C}$ clips each
parameter to the ranges
\begin{equation}\label{eq:intensity_parameter_bounds}
10^{-6}\le\nu_k\le50,
\qquad
10^{-8}\le\alpha_{s,k}\le10,
\qquad
10^{-3}\le\beta_{s,k}\le200.
\end{equation}

At iteration $r$, one calibration day is selected uniformly.
Its history states $R_m^{(r)}=(R_{m,s,k}(\boldsymbol{\beta}^{(r)}))_{s,k}$
are recomputed using the current parameters. For $m<N_T$, the
sampling weights are
\begin{equation}\label{eq:proxy_weights}
w_m^{(r)}
=\left|\lambda_{E_{m+1}}^{(m)}
  (\Delta_m;\boldsymbol{\psi}^{(r)})\right|
 +\varepsilon\sum_{s,k}
  \left|B_{s,k}^{(m)}(\boldsymbol{\psi}^{(r)})\right|,
\qquad \varepsilon=10^{-3}.
\end{equation}
The terminal index, whose loss contribution is zero, receives the
mean of the other weights. Normalization gives
\begin{equation}\label{eq:record_sampling_probabilities}
w_{N_T}^{(r)}=\frac{1}{N_T-1}\sum_{m=1}^{N_T-1}w_m^{(r)},
\qquad
p_m^{(r)}=\frac{w_m^{(r)}}{\sum_{\ell=1}^{N_T}w_\ell^{(r)}}.
\end{equation}
These weights emphasize records with high fitted intensity and
large accumulated excitation.

A set $\mathcal{I}_r$ of $S$ indices is drawn without replacement
using the relative weights $p_m^{(r)}$. Let
$F_m(\boldsymbol{\psi};R_m^{(r)})$ denote the closed-form
contribution $f_m$ evaluated with the history state held fixed.
The sampled objective is
\begin{equation}\label{eq:sampled_lse_objective}
\widetilde{\mathcal{R}}_r(\boldsymbol{z})
=\frac{1}{S T_{\mathrm{obs}}}
 \sum_{m\in\mathcal{I}_r}
 \frac{F_m\!\left(\mathcal{C}(\exp\boldsymbol{z});R_m^{(r)}\right)}
      {p_m^{(r)}}.
\end{equation}
The update uses
$g^{(r)}=\nabla_{\boldsymbol{z}}
\widetilde{\mathcal{R}}_r(\boldsymbol{z}^{(r)})$, with both
the states and sampling weights fixed during differentiation.
In particular, the decay-rate derivatives omit the dependence
of $R_m$ on $\boldsymbol{\beta}$. Also, $p_m^{(r)}$ is a
normalized sampling weight, not the probability that index $m$
is included in the set. The inverse weighting is therefore not
generally unbiased under this sampling design. Together, these
choices define an approximate update for the least-squares
calibration.

The vector $g^{(r)}$ is clipped to a Euclidean norm of at most one
before an Adam update of $\boldsymbol{z}$~\cite{kingma2014adam}.
We use moment coefficients $(0.9,0.999)$ and a denominator
regularizer of $10^{-8}$. All components are initialized at
$\nu_k=0.05$, $\alpha_{s,k}=0.01$, and $\beta_{s,k}=1$.
The reported calibration uses $S=64$ and $300$ updates; its
learning-rate parameter remains $5\times10^{-3}$ throughout.
The final parameters
$\widehat{\boldsymbol{\psi}}
=\mathcal{C}(\exp\boldsymbol{z}^{(300)})$
are retained for intensity reconstruction. The recorded training
loss is the weighted sampled objective evaluated before each
update.


\subsection{Construction of the order-flow inputs}

The calibrated vector $\widehat{\boldsymbol{\psi}}$ is used to
reconstruct the order-flow inputs on the one-second trading grid.
For each trading day, define
\begin{align}
\widehat{R}_{s,e}(t_k^-)
&=\sum_{\substack{n:\,T_n<t_k\\ S_n=s}}
  \exp\!\bigl[-\widehat{\beta}_{s,e}(t_k-T_n)\bigr],
  \label{eq:grid_history_reconstruction}\\
\widehat{\lambda}_e(t_k^-)
&=\widehat{\nu}_e
  +\sum_{s=0}^{d_{\mathrm{src}}-1}
   \widehat{\alpha}_{s,e}\widehat{R}_{s,e}(t_k^-),
  \qquad e=0,\ldots,11.
  \label{eq:grid_intensity_reconstruction}
\end{align}
The history is initialized at zero at the start of each day.
Only records with timestamps strictly earlier than $t_k$
contribute, so both records at $t_k$ are excluded.

The inputs associated with buy and sell order flow are
\begin{equation}\label{eq:side_intensity_inputs}
\lambda_{t_k}^{+}
=\sum_{e=0}^{5}\widehat{\lambda}_e(t_k^-),
\qquad
\lambda_{t_k}^{-}
=\sum_{e=6}^{11}\widehat{\lambda}_e(t_k^-).
\end{equation}
These quantities aggregate the fitted intensities of the
count-bin records, including the zero-count types $e=0$ and
$e=6$. They enter the executed-volume model as order-flow
covariates. The reconstructed sequences are held fixed during
calibration of the Prony coefficients and subsequent strategy
evaluation.

\section{Numerical Validation} 


\subsection{Data Profile}

\paragraph{Asset and sampling window.}
We use LOBSTER-format, millisecond-stamped MSFT event messages
from November~3--7, 2014, together with the first ten levels of the
bid and ask books. The records contain timestamps, message types,
order sizes, prices, and side indicators. We restrict the sample
to the continuous-auction window 09:30--16:00, excluding auctions,
pre- and post-market trading, and overnight periods. The sample
contains approximately $5.4\times10^6$ messages over five sessions
of $T_{\text{day}}=23{,}400$ seconds each. Order modifications,
cancellations, and trading-halt indicators are not used.

\paragraph{Order classification.}
Using the message types, we identify market orders $\mathcal{M}$
as visible and hidden executions against standing limit orders,
and limit orders $\mathcal{L}$ as new submissions. Market-order
direction follows the side consumed: $B_n=0$ denotes a buy order
consuming ask liquidity, and $B_n=1$ a sell order consuming bid
liquidity. The five-day sample contains
$|\mathcal{M}|\approx1.30\times10^5$ market orders and
$|\mathcal{L}|\approx2.65\times10^6$ limit orders, a ratio of
approximately $4.9\%$.

\paragraph{Average order size.}
Per-event sizes $V_n$ are right-skewed and concentrated at multiples
of $100$ shares. We estimate $q_{\text{avg}}^\pm$ by the sample mean,

\begin{equation}\label{eq:qavg_def}
\widehat{q}_{\text{avg}}^\pm
\;:=\;
\frac{1}{|\mathcal{M}^\pm|}\sum_{n\in\mathcal{M}^\pm}V_n,
\end{equation}

where $\mathcal{M}^\pm$ denotes the buy/sell sub-stream.
The five-day averages are $\widehat{q}_{\text{avg}}^+\approx203.6$
and $\widehat{q}_{\text{avg}}^-\approx206.9$ shares, differing by
approximately $1.6\%$.

\paragraph{Market-order arrival profile.}
Table~\ref{tab:mo_profile_summary} summarizes the per-second counts
over $K_{\text{total}}=117{,}000$ intervals, and
Figure~\ref{fig:mo_distribution} displays their distributions.
Most intervals contain no market orders, while active intervals
exhibit a pronounced right tail.

\begin{table}[t]
\centering
\begin{tabular}{lrr}
\toprule
& Buy ($N_{t_k}^+$) & Sell ($N_{t_k}^-$) \\
\midrule
Intervals with $N=0$ &
$105{,}795$ ($90.42\%$) & $106{,}245$ ($90.81\%$) \\
Intervals with $N\ge 1$ &
$\phantom{0}11{,}205$ ($\phantom{0}9.58\%$) & $\phantom{0}10{,}755$ ($\phantom{0}9.19\%$) \\
Mean (all intervals) &
$0.634$ & $0.480$ \\
Mean (active intervals, $N\ge 1$) &
$6.62$ & $5.22$ \\
Maximum observed per-interval count &
$337$ & $91$ \\
\bottomrule
\end{tabular}
\caption{Empirical profile of the per-second market-order counts
$N_{t_k}^\pm$ over the five-day calibration sample
($K_{\text{total}}=117{,}000$ intervals, $\Delta t=1\,\mathrm{s}$).
The active-interval mean is approximately ten times the
all-interval mean on both sides, and the maximum per-interval count
on the buy side exceeds the active-interval mean by a factor
exceeding fifty.}
\label{tab:mo_profile_summary}
\end{table}

\begin{figure}[t]
\centering
\includegraphics[width=0.85\linewidth]{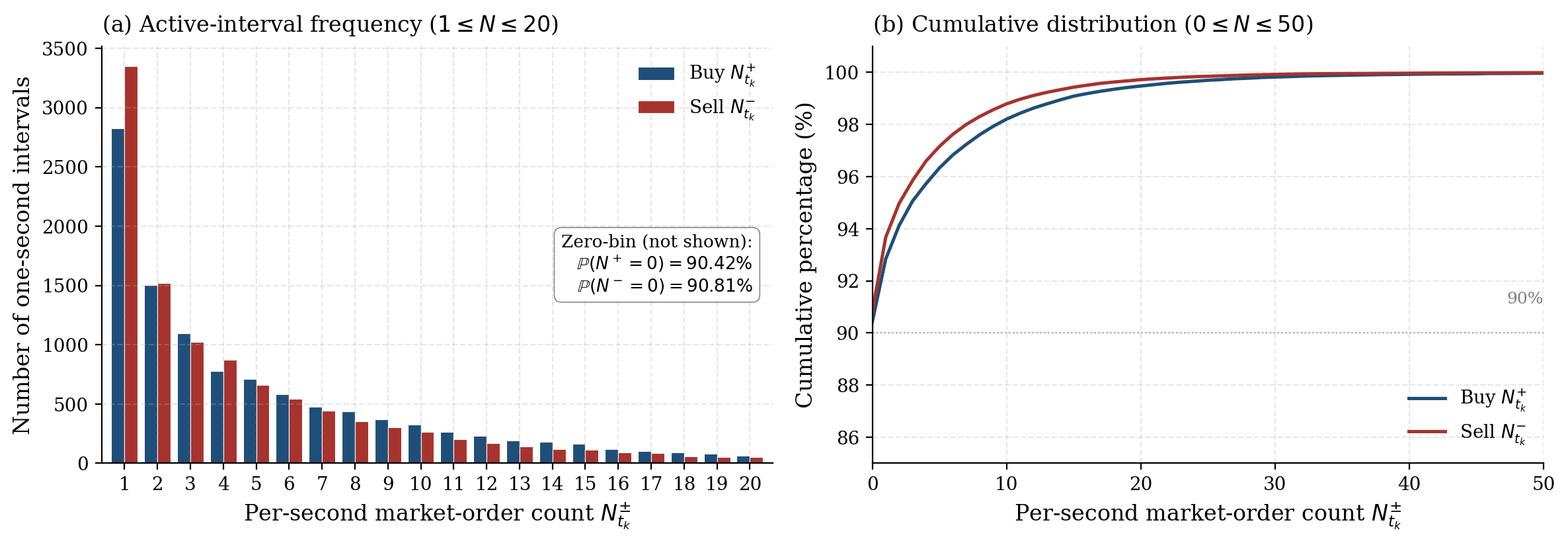}
\caption{Empirical distribution of the per-second market-order counts
$N_{t_k}^\pm$ over the five-day calibration sample. \emph{Left:}
frequency histogram of $N_{t_k}^\pm$ on $\Delta t=1\,\mathrm{s}$,
truncated at $N\le 20$ for visibility; the zero bin accounts for
$90.42\%$/$90.81\%$ of intervals on the buy/sell side. \emph{Right:}
cumulative distribution function, capped at $N\le 50$.}
\label{fig:mo_distribution}
\end{figure}

\paragraph{In-sample/out-of-sample split.}
We use November~3--6 for training
($\mathcal{D}_{\text{train}}$, $K_{\text{train}}=93{,}600$ intervals)
and November~7 for out-of-sample evaluation
($\mathcal{D}_{\text{OOS}}$, $K_{\text{OOS}}=23{,}400$ intervals).
The demand coefficients $(A_N^\pm,B_N^\pm,C_N^\pm)$, sdHawkes
parameters $(\widehat{\boldsymbol{\phi}},\widehat{\boldsymbol{\psi}})$,
and Prony coefficients
$\theta^\star=\arg\max_\theta J(\theta;\mathcal{D}_{\text{train}})$
are estimated on $\mathcal{D}_{\text{train}}$ and held fixed during
out-of-sample evaluation.

\subsection{Exponential Demand Function Results}

We calibrate the exponential demand function:

\begin{equation}\label{eq:exp_demand}
P_N^\pm(L)\;=\;A_N^\pm\,e^{-B_N^\pm\,L}\,+\,C_N^\pm,
\qquad L\in\{1,2,3,4,5\},
\end{equation}

separately for each of the five nonzero count bins and each
direction.

\paragraph{Fit target and objective.}
Let $\mathcal{M}_i^\pm$ denote the market-order events of side
$\pm$ in training interval $i$ (Section~5.1). We retain intervals
with $|\mathcal{M}_i^\pm|\ge 1$ and complete book information
through level~$5$. For a representative order of size
$q_{\text{avg}}^\pm$ posted at the start of the interval, the
front-of-queue fill probability at level $L$ is

\begin{equation}\label{eq:emp_fill}
P_i^\pm(L)\;=\;
\min\!\left\{
\frac{\bigl(S_i^\pm-D_{i,L-1}^\mp\bigr)^{+}}{q_{\text{avg}}^\pm},
\;1
\right\},
\qquad L\in\{1,\ldots,5\},
\end{equation}

where $S_i^\pm:=\sum_{n\in\mathcal{M}_i^\pm}V_n$ is the
aggregate market-order size and $D_{i,L}^\mp$ is the cumulative
opposite-side depth through level $L$, measured at the first event
of $\mathcal{M}_i^\pm$. We set $D_{i,0}^\mp:=0$ and
$(x)^+:=\max(x,0)$. The empirical target is the bin average

\begin{equation}\label{eq:bin_avg}
\widehat{P}_{\mathcal{N}}^{\pm}(L)\;:=\;
\frac{1}{|\mathcal{I}_{\mathcal{N}}^\pm|}
\sum_{i\in\mathcal{I}_{\mathcal{N}}^\pm}P_i^\pm(L),
\qquad
\mathcal{I}_{\mathcal{N}}^\pm\;:=\;
\bigl\{i:|\mathcal{M}_i^\pm|\in\mathcal{N}\bigr\},
\end{equation}

The coefficients are estimated by constrained nonlinear least
squares over the five price levels:

\begin{equation}\label{eq:fit_obj}
(\widehat{A}_{\mathcal{N}}^\pm,\widehat{B}_{\mathcal{N}}^\pm,\widehat{C}_{\mathcal{N}}^\pm)\;=\;
\argmin_{(A,B,C)\in[0,1]\times[0,5]\times[0,1]}
\sum_{L=1}^{5}
\Bigl[\widehat{P}_{\mathcal{N}}^{\pm}(L)\;-\;\bigl(A\,e^{-BL}+C\bigr)\Bigr]^2.
\end{equation}

\paragraph{Calibrated coefficients and goodness of fit.}
Table~\ref{tab:exp_demand_coefs} reports the estimated coefficients
and sample sizes; Figure~\ref{fig:exp_demand} compares the fitted
curves with the empirical targets.

\begin{table}[t]
\centering
\renewcommand{\arraystretch}{1.05}
\begin{tabular}{lrr|rrr|rrr}
\toprule
& \multicolumn{2}{c|}{Sample sizes}
& \multicolumn{3}{c|}{Buy ($P_N^{+}$)}
& \multicolumn{3}{c}{Sell ($P_N^{-}$)} \\
$\mathcal{N}$
& $|\mathcal{I}_{\mathcal{N}}^{+}|$ & $|\mathcal{I}_{\mathcal{N}}^{-}|$
& $\widehat{A}_{\mathcal{N}}^{+}$ & $\widehat{B}_{\mathcal{N}}^{+}$ & $\widehat{C}_{\mathcal{N}}^{+}$
& $\widehat{A}_{\mathcal{N}}^{-}$ & $\widehat{B}_{\mathcal{N}}^{-}$ & $\widehat{C}_{\mathcal{N}}^{-}$ \\
\midrule
$\{1,2,3\}$
& $5{,}426$ & $5{,}895$
& $1.0000$ & $0.7234$ & $1.78\!\times\!10^{-16}$
& $1.0000$ & $0.7353$ & $7.79\!\times\!10^{-21}$ \\
$\{4,5,6\}$
& $2{,}069$ & $2{,}077$
& $1.0000$ & $0.4970$ & $7.07\!\times\!10^{-15}$
& $1.0000$ & $0.4856$ & $6.51\!\times\!10^{-12}$ \\
$\{7,8,9\}$
& $1{,}283$ & $1{,}102$
& $1.0000$ & $0.4596$ & $1.62\!\times\!10^{-12}$
& $1.0000$ & $0.4510$ & $8.77\!\times\!10^{-4}$ \\
$\{10,11,12\}$
& $\phantom{0}819$ & $\phantom{0}640$
& $1.0000$ & $0.4440$ & $6.16\!\times\!10^{-3}$
& $1.0000$ & $0.4441$ & $5.09\!\times\!10^{-3}$ \\
$\{13,14,\ldots\}$
& $1{,}608$ & $1{,}040$
& $1.0000$ & $0.4231$ & $3.27\!\times\!10^{-2}$
& $1.0000$ & $0.4155$ & $4.41\!\times\!10^{-2}$ \\
\bottomrule
\end{tabular}
\caption{Calibrated exponential demand-function coefficients
$(\widehat{A}_{\mathcal{N}}^\pm,\widehat{B}_{\mathcal{N}}^\pm,\widehat{C}_{\mathcal{N}}^\pm)$
of~\eqref{eq:exp_demand} on the four-day training sample
$\mathcal{D}_{\text{train}}$, per $\Delta t=1\,\mathrm{s}$
market-order count bin $\mathcal{N}\in\mathcal{B}\setminus\{0\}$ and
per direction $\pm\in\{+,-\}$. Sample sizes
$|\mathcal{I}_{\mathcal{N}}^{\pm}|$ are the numbers of training-set
intervals contributing to~\eqref{eq:bin_avg}. The $N=0$ bin is
omitted from the table: by construction
$P_{N=0}^\pm(L)\equiv 0$ for all $L$.}
\label{tab:exp_demand_coefs}
\end{table}

\begin{figure}[t]
\centering
\includegraphics[width=\linewidth]{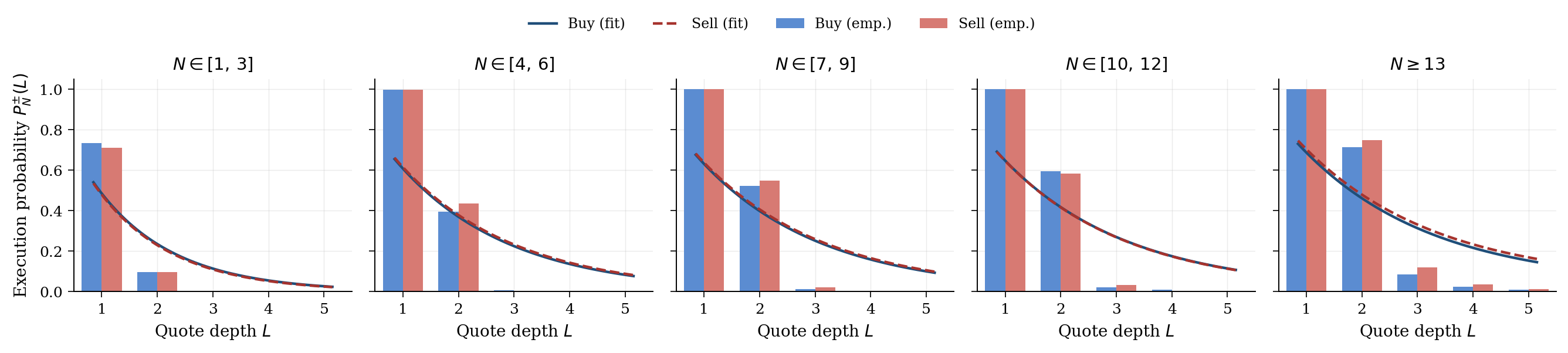}
\caption{Calibrated exponential demand function $P_N^\pm(L)$ by
market-order count bin $\mathcal{N}\in\mathcal{B}\setminus\{0\}$.
Each panel shows the empirical per-bin average
$\widehat{P}_{\mathcal{N}}^{\pm}(L)$ of~\eqref{eq:bin_avg} as bars
(blue: buy side, red: sell side) at $L\in\{1,2,3,4,5\}$, with the
corresponding three-parameter exponential fit~\eqref{eq:exp_demand}
overlaid as curves (solid blue: buy fit, dashed red: sell fit) on
the continuous $L$-axis.}
\label{fig:exp_demand}
\end{figure}

\paragraph{Discussion of the calibrated profile.}
The amplitude constraint is active in every bin and direction,
with $\widehat{A}_{\mathcal{N}}^\pm=1$. The decay rate
$\widehat{B}_{\mathcal{N}}^\pm$ decreases across successively
higher count bins, consistent with more intense market-order flow
consuming liquidity at deeper levels of the book.

The residual constant $\widehat{C}_{\mathcal{N}}^\pm$ is close to
zero in the lower count bins and increases in the highest bin,
reflecting a greater probability of execution at deeper levels.
The fitted profiles are similar across the buy and sell sides.

\subsection{State-Dependent Hawkes Parameters Results}

We estimate the sdHawkes parameters
$\boldsymbol{\psi}=(\boldsymbol{\mu},\boldsymbol{\alpha},\boldsymbol{\beta})$
on $\mathcal{D}_{\text{train}}$ using ASLSD.

\paragraph{Event encoding and state discretization.}
Using the count bins of Section~2.2, each one-second interval
generates two synchronous events at $t_k$, with types
$e_k^+=\iota(N_{t_k}^+)\in\{0,\ldots,5\}$ and
$e_k^-=\iota(N_{t_k}^-)+6\in\{6,\ldots,11\}$, where $\iota$
maps counts to bins. The queue imbalance~\eqref{eq:qi_def} is
discretized at thresholds $-0.6,-0.2,0.2,0.6$ into the five
states, indexed by $x\in\{0,\ldots,4\}$ in
increasing order. Both events inherit $X(t_k^-)$, giving
$d_e=12$ event types and $d_x=5$ states.

\paragraph{Convergence diagnostic.}
The rolling-mean LSE decreases from $-0.23$ to $-1.35$ over
$300$ iterations at a constant learning rate
(Figure~\ref{fig:sdh_convergence}).

\begin{figure}[t]
\centering
\includegraphics[width=0.98\linewidth]{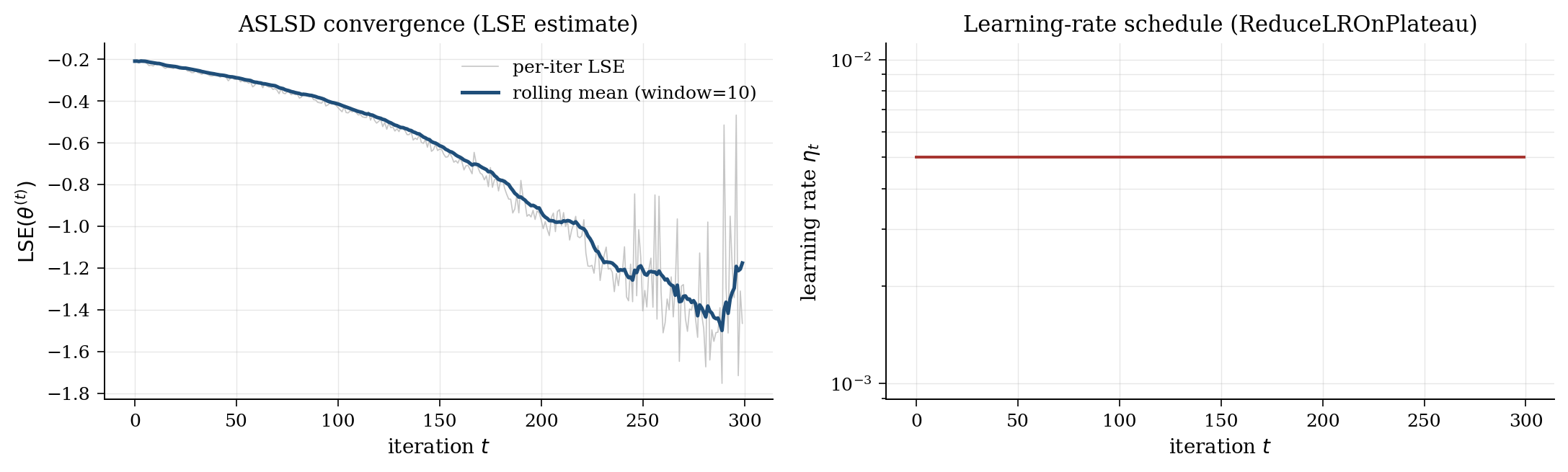}
\caption{ASLSD training trajectory on the four-day training
sample $\mathcal{D}_{\text{train}}$. \emph{Left:} per-iteration
LSE estimate $\widetilde{\mathcal{L}}(\boldsymbol{\psi}^{(t)})$
(grey) and rolling mean with window length~$10$ (dark blue).
\emph{Right:} learning-rate schedule (constant; the plateau
scheduler did not trigger).}
\label{fig:sdh_convergence}
\end{figure}

\paragraph{Baseline intensities.}
The $N=0$ baselines dominate on both sides
(Figure~\ref{fig:sdh_baselines}), consistent with the prevalence
of intervals without market orders. Active-bin
baselines are approximately constant across bins and similar
across directions, leaving their differences in temporal dynamics
to the excitation kernels.

\begin{figure}[t]
\centering
\includegraphics[width=0.98\linewidth]{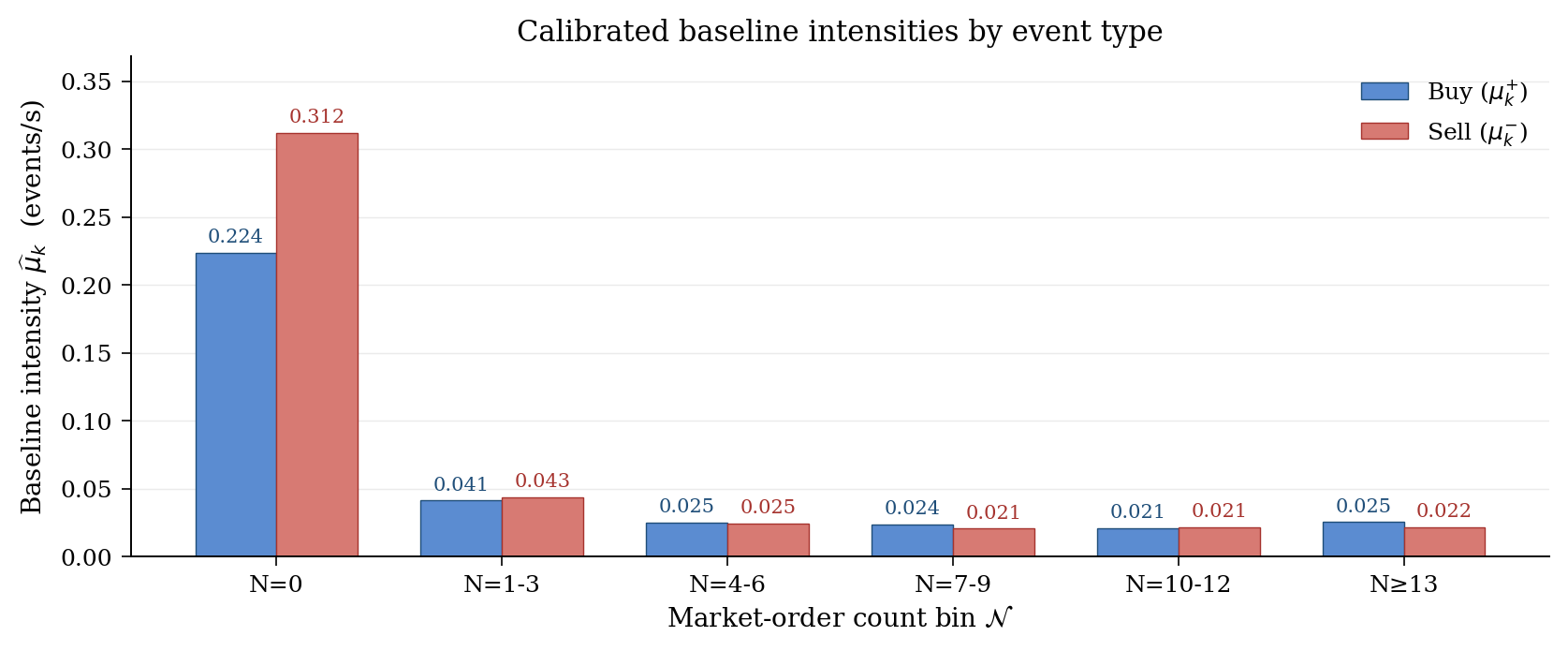}
\caption{Calibrated baseline intensities $\widehat{\mu}_k$, by
market-order count bin $\mathcal{N}\in\mathcal{B}$ and direction.
Buy and sell baselines within each $N\ge 1$ bin agree to within
$3\times 10^{-3}$ events/s; the $N=0$ baselines exceed the
active-bin baselines by roughly one order of magnitude.}
\label{fig:sdh_baselines}
\end{figure}

\paragraph{Excitation magnitude and decay structure.}
Figures~\ref{fig:sdh_alpha} and~\ref{fig:sdh_beta} show the
state averages
$\bar{\alpha}_{j,k}:=d_x^{-1}\sum_x\widehat{\alpha}_{(j,x),k}$
and
$\bar{\beta}_{j,k}:=d_x^{-1}\sum_x\widehat{\beta}_{(j,x),k}$.
Excitation magnitudes are largest for the $N=0$ source and target
types, $j,k\in\{0,6\}$.

Decay rates range from $0.15$ to $2.62\,\mathrm{s}^{-1}$.
Interactions between $N=0$ types have memory horizons of
approximately $5$--$6\,\mathrm{s}$, consistent with persistent
clusters of silent intervals. Cross-side interactions involving
high-count targets decay more rapidly, with rates near
$2\,\mathrm{s}^{-1}$ and sub-second memory horizons.

\begin{figure}[t]
\centering
\includegraphics[width=0.78\linewidth]{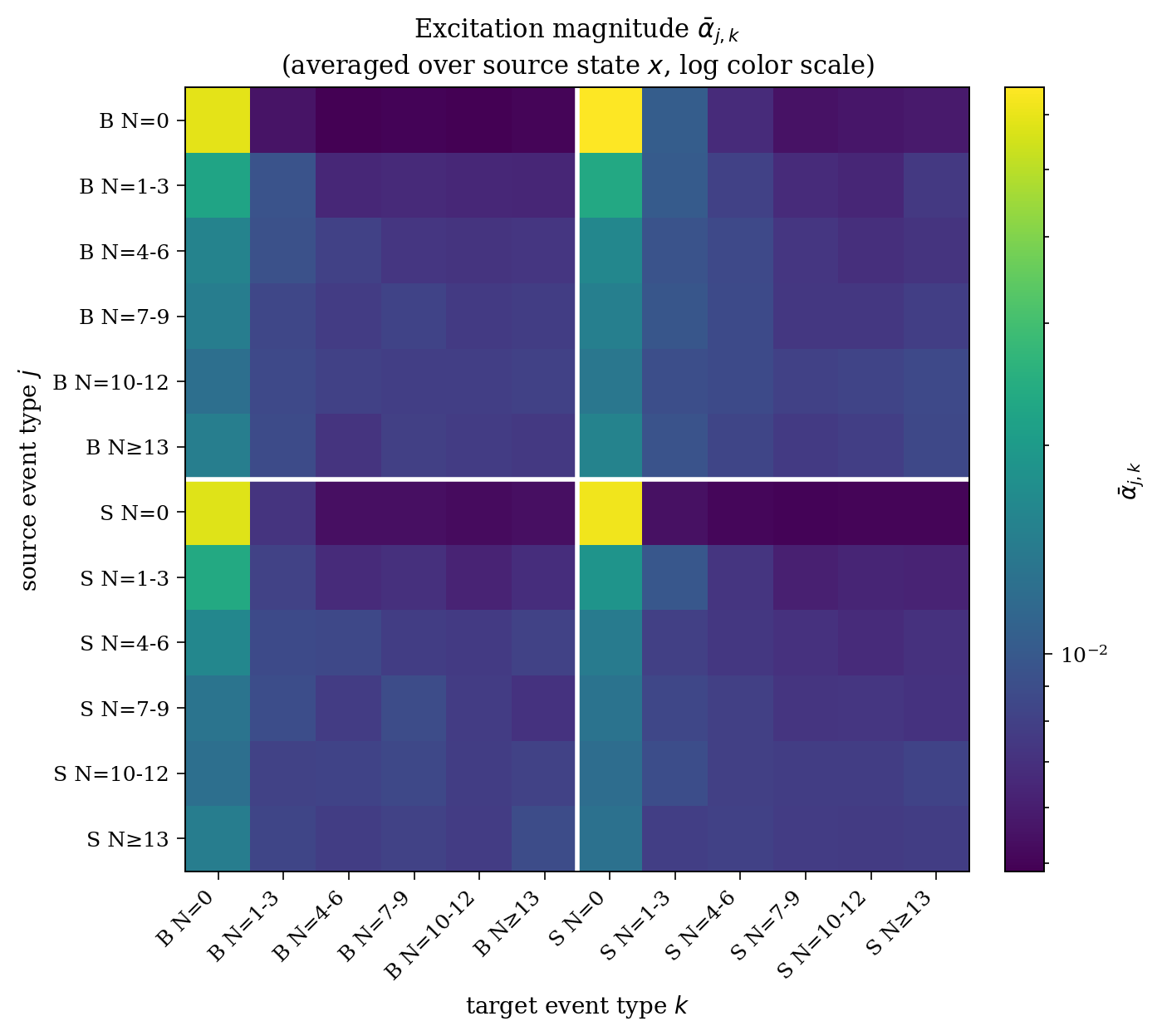}
\caption{Calibrated excitation magnitude $\bar{\alpha}_{j,k}$
of the intensity model of Section~2.2, averaged over source state
$x$. Rows index source event types $j$; columns index target event
types $k$; the white separators delimit the Buy
($k\in\{0,\ldots,5\}$) and Sell ($k\in\{6,\ldots,11\}$) blocks.
The colour scale is logarithmic to accommodate the
three-orders-of-magnitude dynamic range of
$\widehat{\boldsymbol{\alpha}}$.}
\label{fig:sdh_alpha}
\end{figure}

\begin{figure}[t]
\centering
\includegraphics[width=0.78\linewidth]{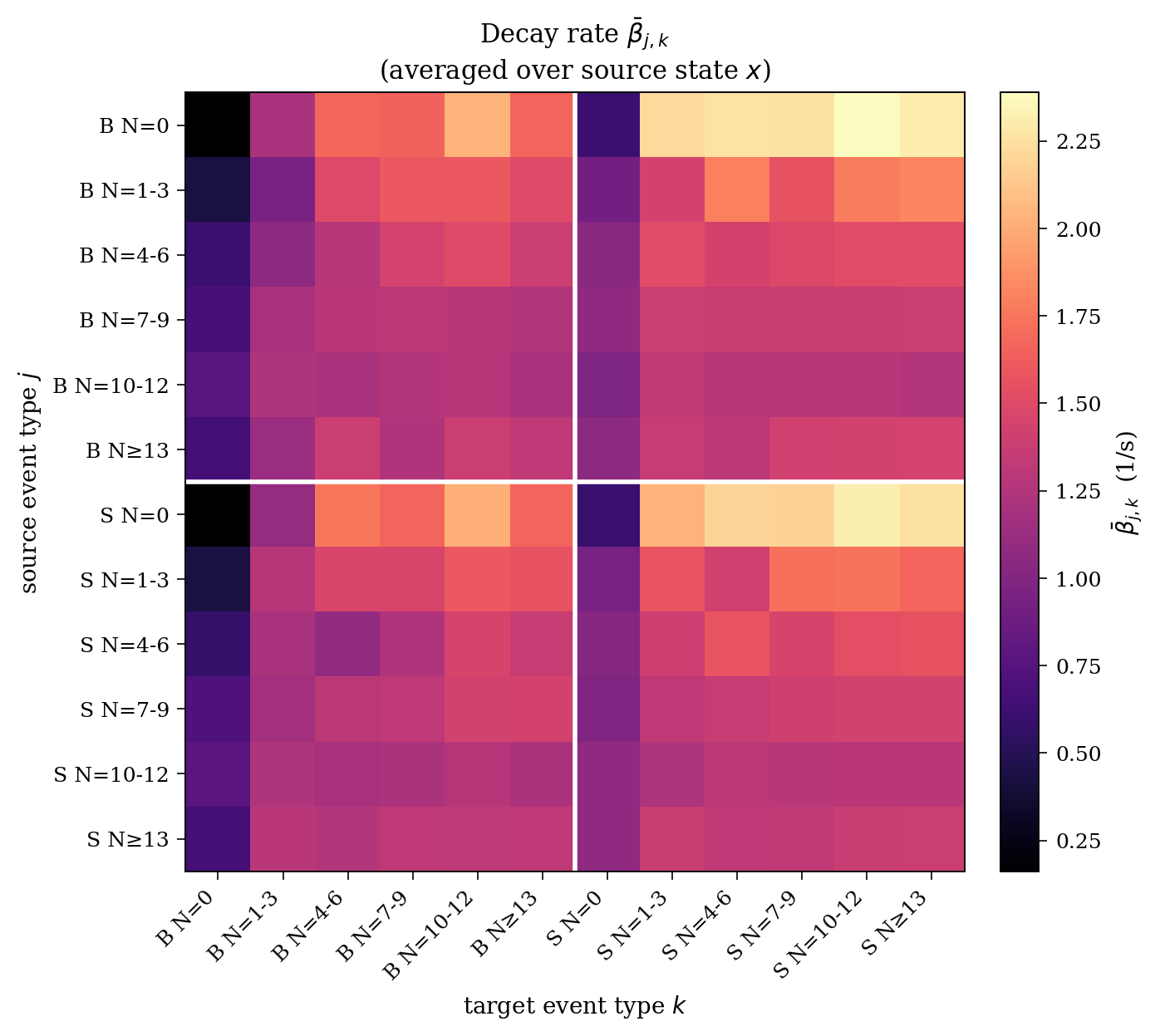}
\caption{Calibrated decay rate $\bar{\beta}_{j,k}$ of the intensity
model of Section~2.2, averaged over source state $x$. The unit of
the colour scale is $\mathrm{s}^{-1}$; the inverse
$\bar{\beta}^{-1}$ has the interpretation of the memory horizon of
the exponential excitation kernel.}
\label{fig:sdh_beta}
\end{figure}

\paragraph{Branching stability.}
Figure~\ref{fig:sdh_stability} reports the per-state self-excitation
ratios $\widehat{\alpha}_{(e,x),e}/\widehat{\beta}_{(e,x),e}$
\citep{hawkes1971,bremaud1996stability}. Their maximum is $0.42$
for buy $N=0$ in the neutral state; ratios for active bins are
at most $0.02$. These diagonal ratios satisfy the necessary
condition of being below unity. The sufficient condition concerns
the spectral radius of the full per-state matrix
$[\widehat{\alpha}_{(j,x),k}/\widehat{\beta}_{(j,x),k}]_{j,k}$.
With off-diagonal contributions negligible relative to the diagonal
terms (Figures~\ref{fig:sdh_alpha} and~\ref{fig:sdh_beta}), the
calibrated process satisfies this condition and is stationary.

\begin{figure}[t]
\centering
\includegraphics[width=0.75\linewidth]{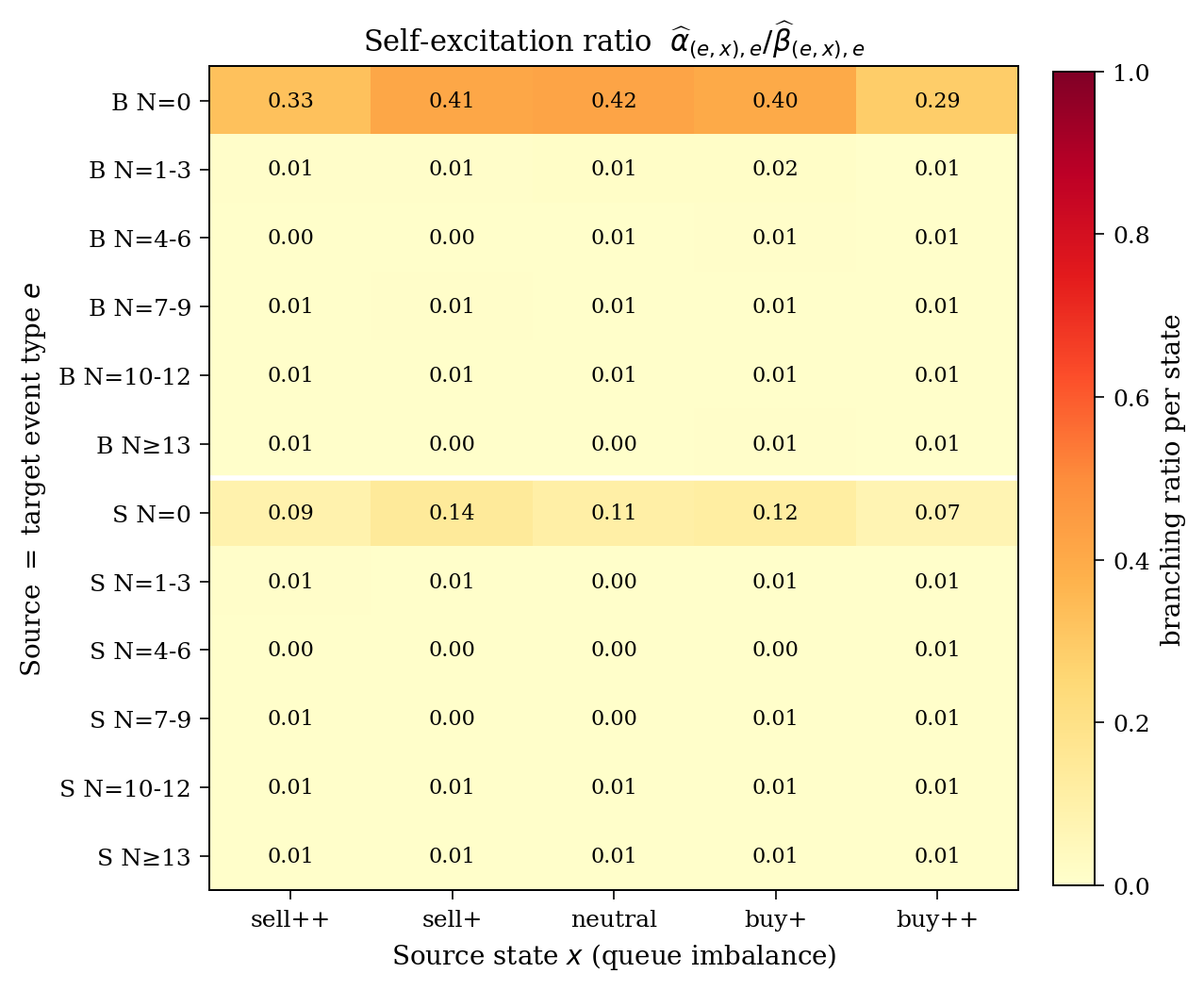}
\caption{Per-state self-excitation ratio
$\widehat{\alpha}_{(e,x),e}/\widehat{\beta}_{(e,x),e}$. Rows index
the event type $e\in\{0,\ldots,11\}$, columns index the
queue-imbalance state $x\in\{0,\ldots,4\}$. All entries are
strictly below unity, satisfying the necessary condition for
stationarity; the largest entries ($0.29$--$0.42$) occur in the
buy $N=0$ row, reflecting the multi-second persistence of silent
buy intervals visible in the slow-decay corner of
Figure~\ref{fig:sdh_beta}.}
\label{fig:sdh_stability}
\end{figure}


\subsection{Dynamical Programming and Parameters Calibration}
\label{sec:cmaes_calibration}

With the demand and sdHawkes parameters fixed, we calibrate the Prony coefficients by the
CMA-ES procedure under a hard inventory bound.

\paragraph{Parametric ansatz.}
In the executed-volume specification~\eqref{eq:Q_form}, we use

\begin{equation}\label{eq:prony_form}
f^\pm(L,N,\lambda;\boldsymbol{\theta})
\;=\;
\sum_{i=1}^{M^\pm}
w_i^\pm\,
\exp\!\bigl(b_{i,1}^\pm L+b_{i,2}^\pm N+b_{i,3}^\pm \lambda\bigr).
\end{equation}

with $M^+=M^-=2$, giving
$\boldsymbol{\theta}=(\boldsymbol{\theta}^+,\boldsymbol{\theta}^-)\in\R^{16}$.

\paragraph{Outer objective and inventory constraint.}
The calibration maximizes the mean value across training days,
starting from zero inventory:

\begin{equation}\label{eq:outer_obj}
V(\boldsymbol{\theta})
\;:=\;
\frac{1}{|\mathcal{D}_{\text{train}}|}
\sum_{d\in\mathcal{D}_{\text{train}}}
h_d\!\bigl(0,0;\,\boldsymbol{\theta}\bigr),
\end{equation}

where $h_d$ is evaluated by the backward recursion of
Proposition~\ref{prop:cartea}, subject to

\begin{equation}\label{eq:hard_bound}
\bigl|I_{t_k}\bigr|\;\le\;I_{\max}\quad\text{for all }k=0,\ldots,K_{\text{day}},
\qquad I_{\max}=10^{4}\text{ shares}.
\end{equation}

Inadmissible continuation states are assigned the barrier value
$\mathfrak{B}=-10^{15}$. Sensitivity to the soft inventory penalty is examined in Section~\ref{subsec:lambda_sensitivity}.

\paragraph{CMA-ES configuration.}
We solve
$\widehat{\boldsymbol{\theta}}:=\arg\max_{\boldsymbol{\theta}\in\Theta_{\mathrm{box}}}V(\boldsymbol{\theta})$
over the parameter bounds

\begin{equation}\label{eq:box}
\Theta_{\mathrm{box}}\;:=\;\bigl(
[0,2]\times[-0.5,0.3]\times[0,0.3]\times[0,3]\bigr)^{\times 4},
\end{equation}

with initial parameters

\begin{equation}\label{eq:warm_start}
\boldsymbol{\theta}^{(0)}
\;:\;
w_1^\pm=w_2^\pm=0.5,\quad
b_{1,2}^\pm=0.05,\quad
b_{2,2}^\pm=0.15,\quad
\text{all other entries zero},
\end{equation}

The population size is $\lambda_{\mathrm{cma}}=12$, the initial
step size is $\sigma^{(0)}=0.1$
Each run allows at most $25$ generations, with termination
tolerance $\Delta V<10^3$ tick-shares (\$10/day).

\paragraph{Convergence trajectories.}
Figure~\ref{fig:cmaes_convergence} shows the best-so-far objective
value and mutation step size. The three penalty settings yield
the same calibrated vector $\widehat{\boldsymbol{\theta}}$;
their value sensitivity is reported in
Section~\ref{subsec:lambda_sensitivity}.

\begin{figure}[t]
\centering
\includegraphics[width=0.98\linewidth]{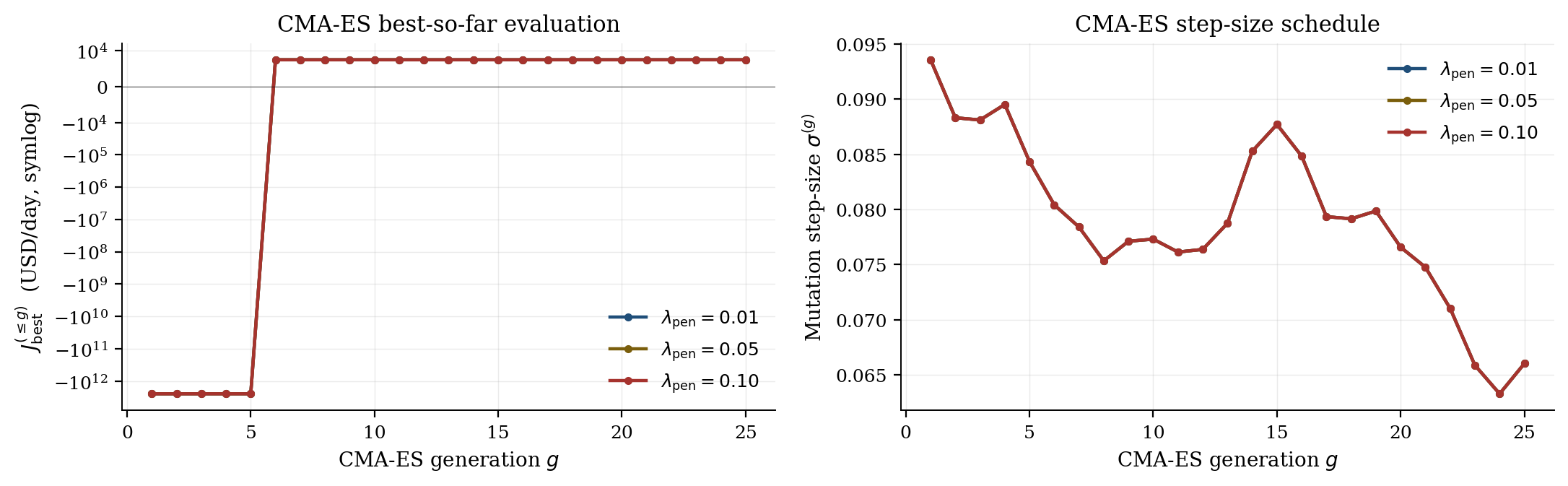}
\caption{CMA-ES convergence and step-size schedule. \emph{Left:}
best-so-far $V_{\mathrm{best}}^{(\le g)}$ in USD/day on a
symmetric-log axis. The first five generations are dominated by
candidates violating~\eqref{eq:hard_bound} and are evaluated at the
barrier level $\mathfrak{B}\approx-\$2.5\times 10^{12}$ per day;
at generation~$6$, the algorithm first samples a feasible
candidate, and $V_{\mathrm{best}}^{(\le g)}$ remains at
$\$7{,}462.80/\mathrm{day}$ thereafter. \emph{Right:} mutation
step-size $\sigma^{(g)}$ shrinks modestly from
$0.094$ to $0.066$ over the run. The three
$\lambda_{\mathrm{pen}}\in\{0.01,0.05,0.10\}$ trajectories
coincide to machine precision.}
\label{fig:cmaes_convergence}
\end{figure}

\paragraph{Calibrated coefficients.}
Figure~\ref{fig:cmaes_theta_bars} compares the calibrated
coefficients with their initial values and bounds. The weights
remain close to $0.5$, with summed amplitudes near unity on each
side. The principal asymmetry lies in the count coefficients:
$\widehat{b}_{i,2}^\pm$ are near zero on the buy side but remain
positive on the sell side. The depth and intensity coefficients
are small relative to their permitted ranges.

\begin{figure}[t]
\centering
\includegraphics[width=0.95\linewidth]{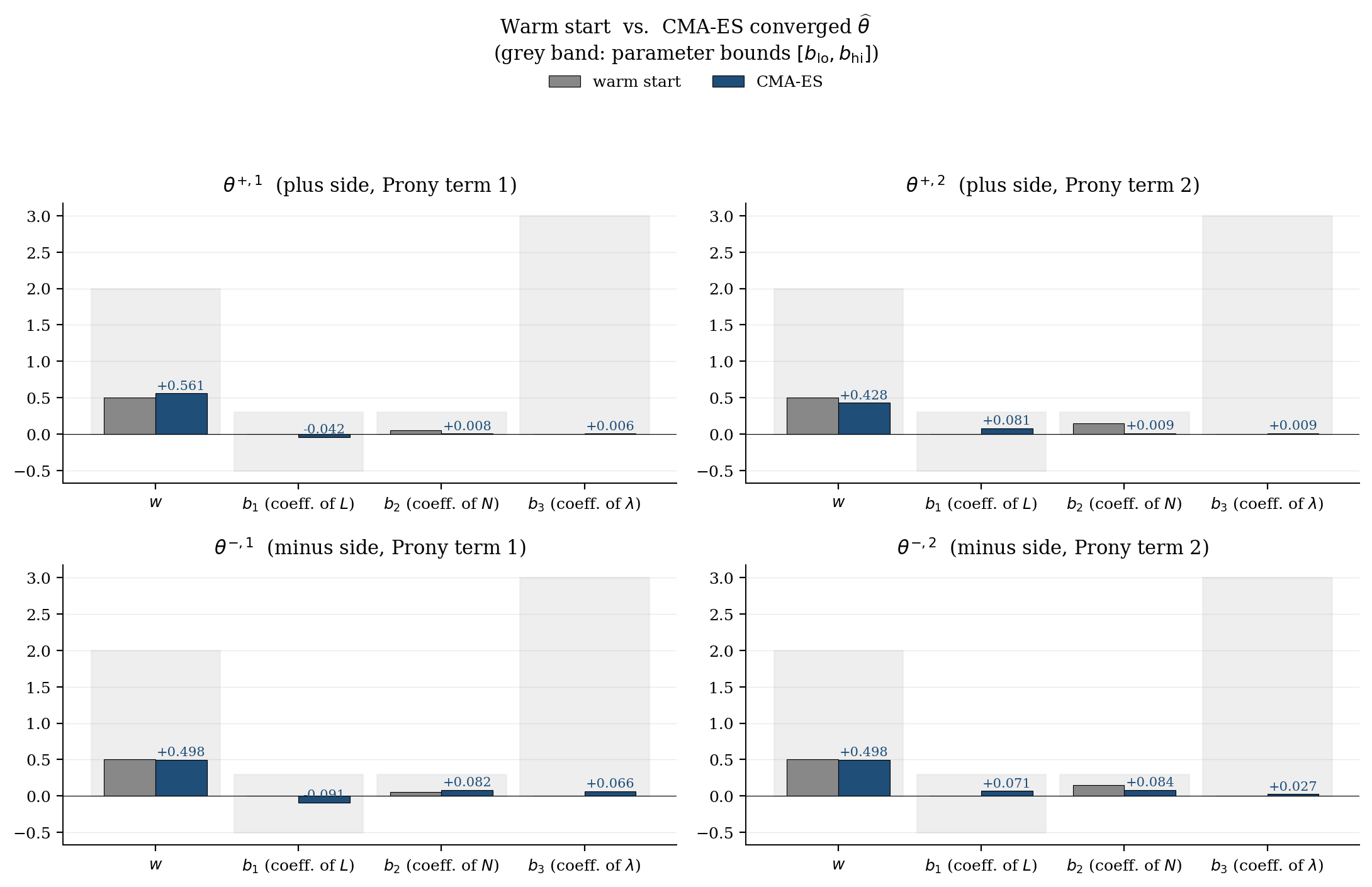}
\caption{CMA-ES converged
$\widehat{\boldsymbol{\theta}}=(\widehat{\boldsymbol{\theta}}^+,\widehat{\boldsymbol{\theta}}^-)$
versus the warm start~\eqref{eq:warm_start}, by side and Prony
term. Grey bands display the parameter
bounds~\eqref{eq:box}; the calibrated values are reported above
each blue bar.}
\label{fig:cmaes_theta_bars}
\end{figure}

The resulting multiplier profiles
(Figure~\ref{fig:cmaes_f_curves}) are nearly constant on the buy
side and increase with market-order count on the sell side.

\begin{figure}[t]
\centering
\includegraphics[width=0.98\linewidth]{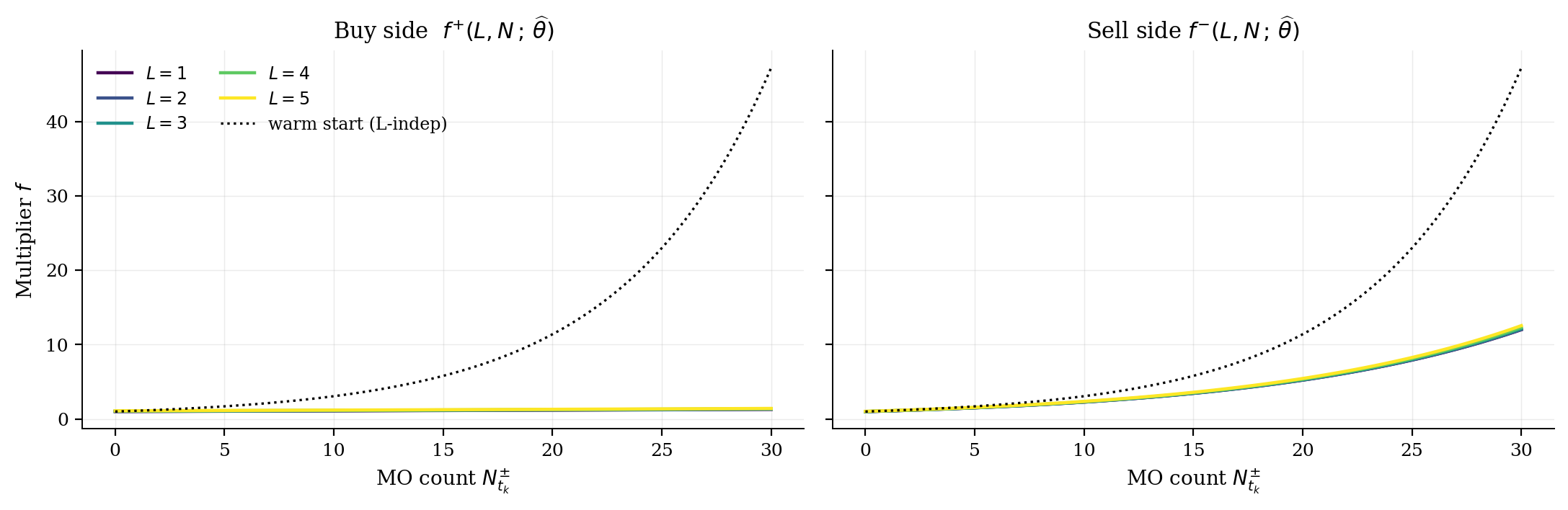}
\caption{Calibrated multiplier
$f^\pm(L,N;\widehat{\boldsymbol{\theta}})$ as a function of the
MO count $N$ at $\lambda=0$, evaluated at
$L\in\{1,2,3,4,5\}$ (colour bar). The dotted black curve is
the warm-start multiplier
$f^\pm(L,N;\boldsymbol{\theta}^{(0)})=
e^{0.05 N}/2+e^{0.15 N}/2$, identical across $L$. The CMA-ES
reduces the buy-side multiplier to near-unity uniformly in
$N$ and retains a moderate $N$-dependence on the sell side.}
\label{fig:cmaes_f_curves}
\end{figure}

\paragraph{In-sample value.}
The calibrated in-sample value is
$V(\widehat{\boldsymbol{\theta}})=\$7{,}462.80/\mathrm{day}$,
compared with $\$4{,}338.39/\mathrm{day}$ for $f\equiv1$ under
the same hard inventory bound, an improvement by a factor of $1.72$.

\paragraph{Local-optimum caveat.}
The buy--sell asymmetry, despite the similar order-flow profiles
in Sections~5.1--5.3, suggests sensitivity to local optima.
Section~5.5 compares this calibration with parameters estimated
without the hard bound and subsequently evaluated under the bound
with the no-post action enabled; the latter achieve a higher
out-of-sample value.

\subsubsection{Sensitivity to the soft inventory penalty $\lambda_{\mathrm{pen}}$}
\label{subsec:lambda_sensitivity}

At $I_{\max}=10{,}000$ shares, we repeat the calibration for
$\lambda_{\mathrm{pen}}\in\{0.01,0.05,0.10\}$.
Table~\ref{tab:lambda_sensitivity} reports the training and
out-of-sample values and inventories. All sixteen calibrated
coefficients agree to numerical precision across the three runs,
while the value differences are small relative to trading revenues.

\begin{table}[ht]
\centering
\renewcommand{\arraystretch}{1.05}
\begin{tabular}{lrrrr}
\toprule
$\lambda_{\mathrm{pen}}$
 & $V_{\mathrm{train}}$
 & $V_{\mathrm{OOS}}$
 & $\max_t |I_t|^{\mathrm{OOS}}$
 & $|I_T|^{\mathrm{OOS}}$ \\
\midrule
$0.01$ & $\$7{,}462.80$ & $\$10{,}516.24$ & $9{,}977.4$ & $77.7$ \\
$0.05$ & $\$7{,}462.53$ & $\$10{,}460.00$ & $9{,}992.3$ & $7.5$ \\
$0.10$ & $\$7{,}462.41$ & $\$10{,}443.89$ & $9{,}997.2$ & $9.6$ \\
\midrule
\emph{range} & $\$0.39$ & $\$72.35$ & $19.8$ & $70.2$ \\
\bottomrule
\end{tabular}
\caption{Sensitivity to $\lambda_{\mathrm{pen}}$ at fixed hard bound
$I_{\max}=10^{4}$ shares. Values are USD/day; inventories are in
shares. The calibrated $\widehat{\boldsymbol{\theta}}$ is identical
across the three rows.}
\label{tab:lambda_sensitivity}
\end{table}

\paragraph{Mechanism.}
The soft penalty acts through the terminal contribution
$-\lambda_{\mathrm{pen}}I_T^2$. The calibrated policies close
positions toward the end of the day, leaving small terminal
inventories and hence a small penalty relative to trading revenues.
Running inventories, by contrast, approach the boundary and are
controlled by the hard inventory constraint.

\subsection{Optimal Size Strategy}
Fix the soft inventory penalty at $\lambda_{\mathrm{pen}}=0.05$ and the
hard inventory bound at $I_{\max}=10{,}000$ shares. Under these settings
we compare the closed-form executed volume
\begin{equation}\label{eq:Q_closed_form}
Q_{t_{k+1}}^{\pm}
\;=\;
q_{\mathrm{avg}}^{\pm}\;\cdot\;
\Bigl[
\sum_{i=1}^{2} \widehat{w}_i^{\pm}\,
\exp\!\bigl(\widehat{b}_{i,1}^{\pm} L_{t_k}^{\pm}
            +\widehat{b}_{i,2}^{\pm} N_{t_k}^{\pm}
            +\widehat{b}_{i,3}^{\pm} \lambda_{t_k}^{\pm}\bigr)
\Bigr]\;\cdot\;
\bigl[\widehat{A}_{N}^{\pm}\,e^{-\widehat{B}_{N}^{\pm} L_{t_k}^{\pm}}
      +\widehat{C}_{N}^{\pm}\bigr]
\end{equation}
against four alternative strategies. The closed-form policy is
admitted under two distinct calibrations of $\widehat{\boldsymbol{\theta}}$:
$\mathcal{S}_{\mathrm{opt}}$ uses the parameters identified by CMA-ES
on the action set $\{1,\ldots,5\}^{2}$ \emph{without} the hard bound
during training, then evaluated under the bound with the no-post
action $L^{\pm}=0$ enabled; $\mathcal{S}_{\mathrm{retr}}$ uses the
hard-bound-constrained calibration of Section~5.4.
The baselines are: $\mathcal{S}_{0}$, the trivial multiplier
$f\equiv 1$; $\mathcal{N}_{1}$, the fixed-quote policy that always
posts at $L^{\pm}=1$ with size $q_{\mathrm{avg}}^{\pm}$ and ignores
the inventory; and $\mathcal{N}_{2}$, the myopic-greedy policy that
selects $(L^{+},L^{-})$ to maximise the one-step reward
$Q^{+}\Delta+Q^{-}\Delta$ at every grid point, again ignoring the
inventory.

\paragraph{Headline result.}
Table~\ref{tab:headline} reports the realised value $V_{\mathrm{OOS}}$
on the out-of-sample trading day, together with the maximum and
terminal inventories.

\begin{table}[ht]
\centering
\renewcommand{\arraystretch}{1.05}
\begin{tabular}{lrrrc}
\toprule
Strategy & $V_{\mathrm{OOS}}$ (USD/day) & $\max_t |I_t|$ (sh.) & $I_T$ (sh.) & feasible\\
\midrule
$\mathcal{S}_{\mathrm{opt}}$  (legacy $\widehat{\boldsymbol{\theta}}$)    & $\mathbf{\$119{,}707}$ & $9{,}988$ & $-16.5$ & $\checkmark$ \\
$\mathcal{S}_{\mathrm{retr}}$  (retrained $\widehat{\boldsymbol{\theta}}$) & $\$10{,}460$ & $9{,}992$ & $+7.5$ & $\checkmark$ \\
$\mathcal{S}_{0}$  (trivial $f\equiv 1$)                                  & $\$8{,}559$ & $9{,}944$ & $-12.4$ & $\checkmark$ \\
$\mathcal{N}_{1}$  (fixed $L=1,\,Q=q_{\mathrm{avg}}$)                     & $-\$24{,}712{,}461$ & $232{,}369$ & $+222{,}627$ & $\times$ \\
$\mathcal{N}_{2}$  (myopic greedy)                                        & $-\$20{,}250{,}905$ & $204{,}802$ & $-201{,}894$ & $\times$ \\
\bottomrule
\end{tabular}
\caption{Out-of-sample comparison on 2014-11-07 under
$(\lambda_{\mathrm{pen}}, I_{\max})=(0.05,\,10{,}000)$. The headline
strategy $\mathcal{S}_{\mathrm{opt}}$ delivers a $11.4\times$ uplift
over the closest alternative $\mathcal{S}_{\mathrm{retr}}$ at a
comparable peak inventory.}
\label{tab:headline}
\end{table}

The closed-form policy with legacy $\widehat{\boldsymbol{\theta}}$
delivers $\$119{,}707/$day, an $11.4\times$ uplift over the closest
admissible competitor and a $14.0\times$ uplift over the trivial
$f\equiv 1$ baseline. The two naïve strategies post negative values
in the order of tens of millions of dollars: although each accumulates
moderate intra-day reward (Figure~\ref{fig:cum_reward}), their
inventories drift to $|I_T|\gtrsim 2\times 10^{5}$
shares (Figure~\ref{fig:inventory_paths}, lower panel) and the
mark-to-market revaluation at the terminal price dominates the
realised value.

\begin{figure}[ht]
\centering
\includegraphics[width=0.98\linewidth]{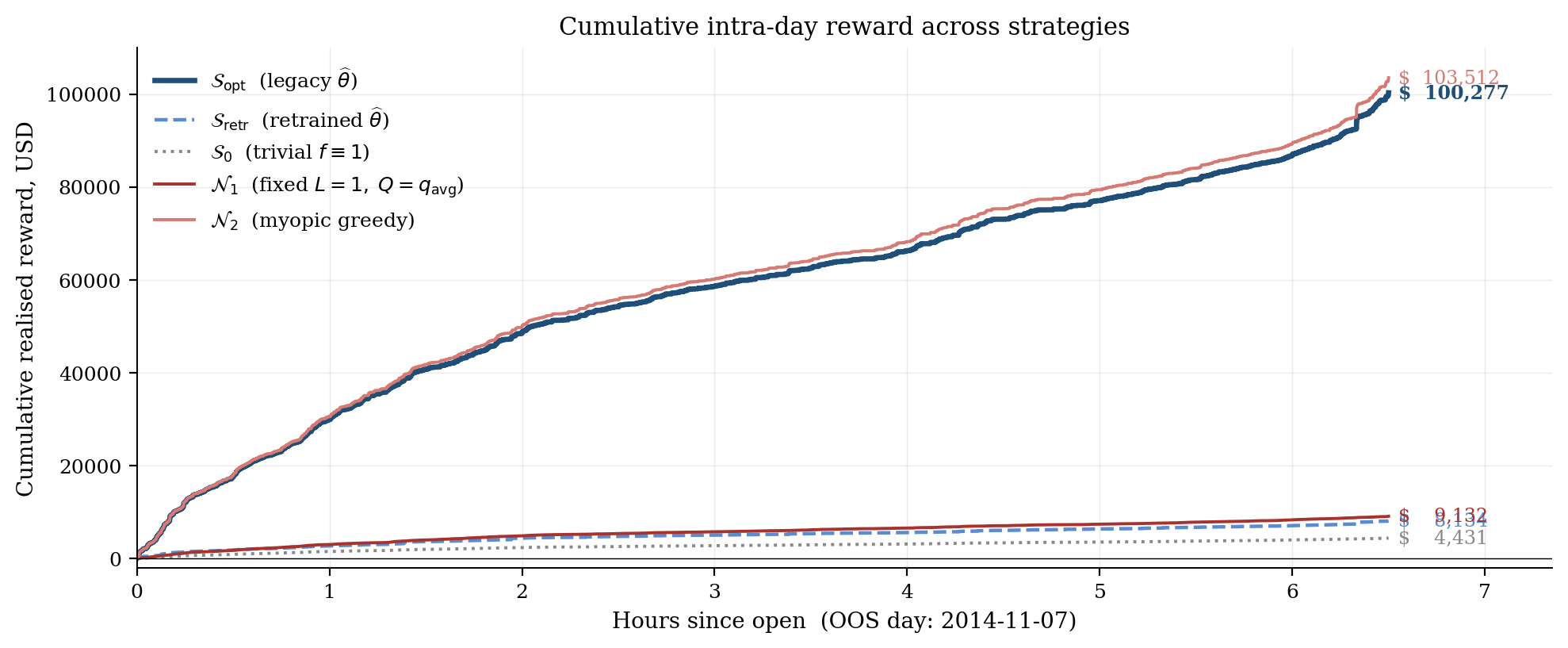}
\caption{Cumulative intra-day reward across the five strategies. The
two closed-form policies and the myopic-greedy baseline accumulate
in the order of $\$10^{5}$ over the trading day; the trivial,
fixed-quote, and retrained policies remain in the order of
$\$10^{4}$.}
\label{fig:cum_reward}
\end{figure}

\begin{figure}[ht]
\centering
\includegraphics[width=0.98\linewidth]{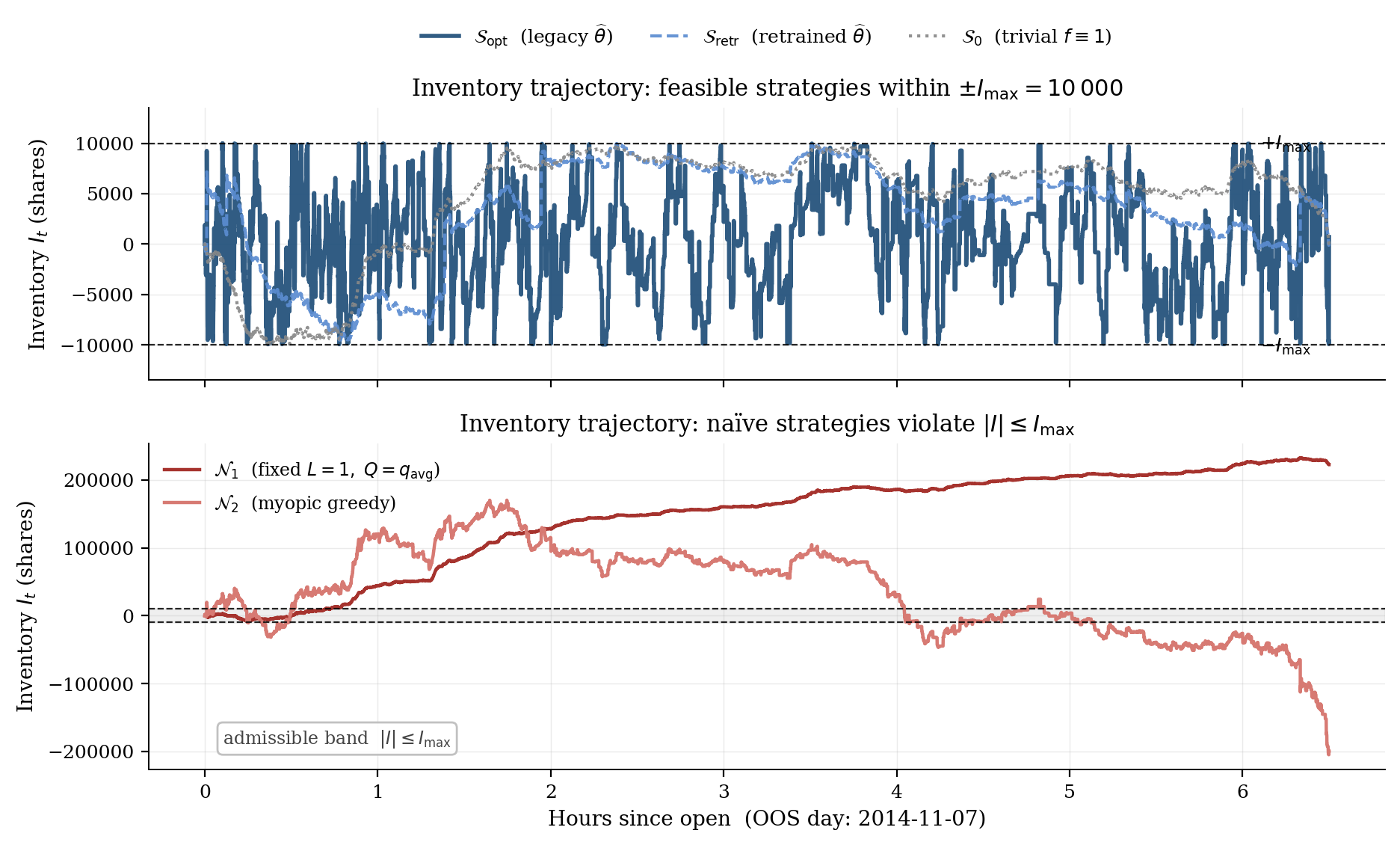}
\caption{Inventory trajectories on the out-of-sample day.
\emph{Top}: the three feasible strategies remain within
$\pm I_{\max}=10{,}000$ shares throughout. \emph{Bottom}: the
fixed-quote $\mathcal{N}_{1}$ and the myopic-greedy
$\mathcal{N}_{2}$ exit the admissible band within the first hour
and reach absolute inventories beyond $2\times 10^{5}$ shares.}
\label{fig:inventory_paths}
\end{figure}

\paragraph{Cumulative reward, inventory, and the closed-form advantage.}
The intra-day cumulative reward of $\mathcal{S}_{\mathrm{opt}}$
($\$100{,}277$) is, in fact, marginally below that of the
myopic-greedy $\mathcal{N}_{2}$ ($\$103{,}512$), since
$\mathcal{N}_{2}$ chases the largest one-step trading revenue at
every grid point. The closed-form policy nevertheless realises an
out-of-sample value four orders of magnitude higher than the myopic
greedy: the difference is the inventory at the close.
$\mathcal{S}_{\mathrm{opt}}$ delivers $|I_T|\le 17$ shares;
$\mathcal{N}_{2}$ delivers $I_T=-201{,}894$ shares, whose
mark-to-market revaluation at the terminal mid-price erases the
intra-day accumulation an order of magnitude over. The closed-form
policy is, in this sense, the unique strategy in the comparison set
that simultaneously achieves the trading-revenue scale of the
myopic policy and the inventory-control discipline of the bound.

\begin{figure}[ht]
\centering
\includegraphics[width=0.98\linewidth]{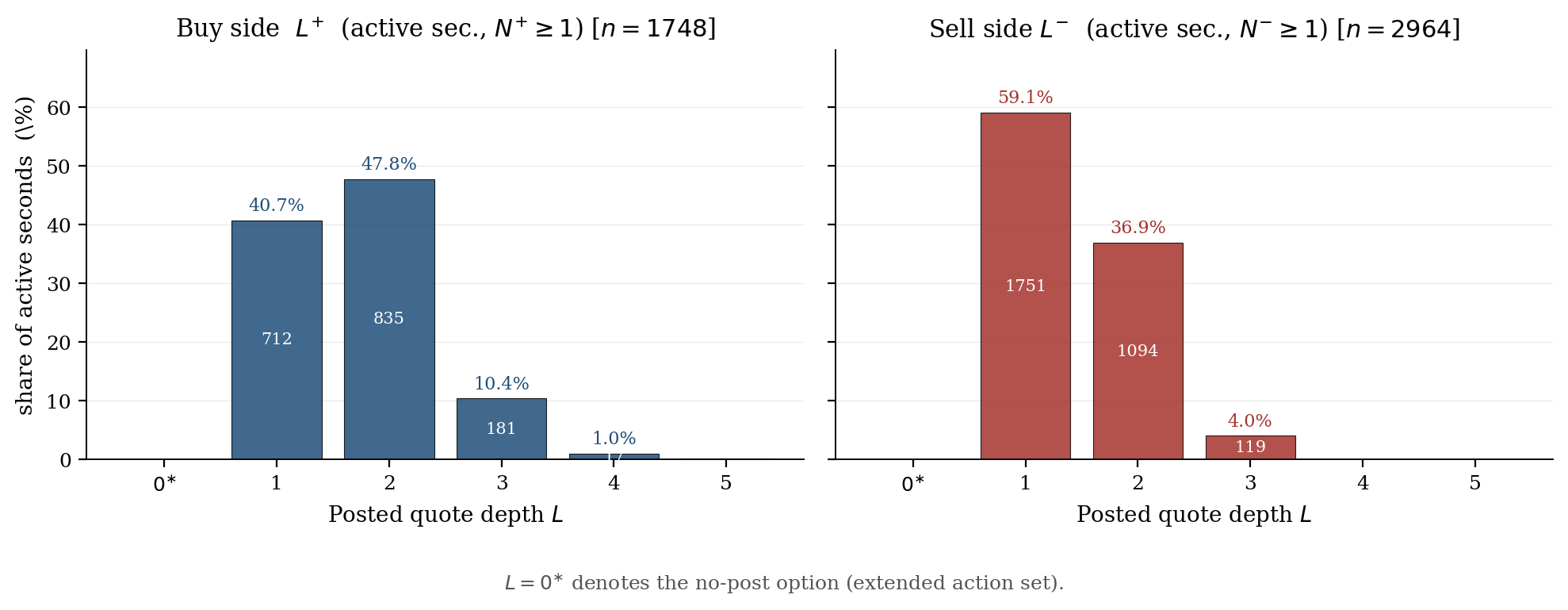}
\caption{Action distribution of $\mathcal{S}_{\mathrm{opt}}$ on the
out-of-sample day, conditional on the side being active
($N^{\pm}\ge 1$). The buy side selects $L^{+}=2$ in $47.8\%$ of
active seconds, with $L^{+}=1$ at $40.7\%$ and the remaining mass
spread over $L^{+}\in\{3,4,5\}$. The sell side concentrates at
$L^{-}=1$ ($59.1\%$) and $L^{-}=2$ ($36.9\%$). The no-post
option $L^{\pm}=0^{\ast}$ is unused on the realised inventory
path.}
\label{fig:action_distribution}
\end{figure}

\paragraph{Source of the advantage.}
Figure~\ref{fig:action_distribution} attributes the
$\mathcal{S}_{\mathrm{opt}}$ vs.\ $\mathcal{S}_{\mathrm{retr}}$ gap
to the buy-side depth-aware quoting that legacy calibration
encodes. Under $\mathcal{S}_{\mathrm{retr}}$ the buy-side quote is
locked at $L^{+}=1$ in essentially all active seconds, an artefact
of the local optimum identified in Section~5.4. Under
$\mathcal{S}_{\mathrm{opt}}$ the modal buy-side quote is
$L^{+}=2$, with non-trivial mass at $L^{+}=3,4,5$ in the seconds
with the largest market-order count: it is precisely on those
high-$N^{+}$ seconds that the closed-form multiplier
$f^{+}(L,N,\lambda;\widehat{\boldsymbol{\theta}})$ chooses to fade
the quote to deeper levels in order to control the inventory
build-up. The sell-side distribution, by contrast, remains
concentrated at $L^{-}\in\{1,2\}$ on both calibrations, consistent
with the directional symmetry of the underlying data: it is the
buy-side depth axis that the legacy calibration learns to use as
an inventory-management instrument.

\medskip
\noindent The no-post option $L^{\pm}=0^{\ast}$, introduced in the extended
action set to render the legacy $\widehat{\boldsymbol{\theta}}$
admissible under the hard bound, is unused on the realised
inventory path of the out-of-sample day; its role is therefore that
of a feasibility safety valve rather than a tactical posting choice.

\section{Conclusions}
This paper derives, for the first time, a closed-form expression for
the expected executed volume under optimal placement in a
state-dependent Hawkes (sdHawkes) limit-order-book market-making
model. Within the absolute-size framework of~\cite{cartea2015book},
we factor the executed volume into three structural components,
\begin{equation}\label{eq:Q_closed_form_concl}
Q_{t_{k+1}}^{\pm}
\;=\;
q_{\mathrm{avg}}^{\pm}\;\cdot\;
f_N^{\pm}\!\big(L_{t_k}^{\pm},\,N_{t_k}^{\pm},\,\lambda_{t_k}^{\pm}\big)\;\cdot\;
P_N^{\pm}\!\big(L_{t_k}^{\pm}\big),
\end{equation}
where $P_N^{\pm}$ is the exponential demand profile fitted block-wise
in the per-second market-order count (Section~2.3), and the size
factor $f_N^{\pm}$ is shown in Section~2.4 to lie in a Banach space of
bounded continuous functions whose elements admit a uniform
Prony-series approximation. The Prony coefficients
$\boldsymbol{\theta}=\{w_i^{\pm},\beta_{ij}^{\pm}\}_{i\le M^{\pm}}$
are calibrated in closed loop with the Cartea-ansatz backward
recursion by a covariance-matrix-adaptation evolution
strategy, bypassing the PIDE route that is customary in
marked point-process control problems but here intractable due to the
state-dependent kernel.

\medskip
\noindent A structural feature of equation~\eqref{eq:Q_closed_form_concl}
deserves emphasis. The closed form of $Q_{t_{k+1}}^{\pm}$ takes the
placement level $L_{t_k}^{\pm}$ as a primitive input and is otherwise
agnostic to its provenance. The dynamic-programming derivation adopted
here is one such device, but any alternative procedure that supplies
an $\big(L_{t_k}^{+},L_{t_k}^{-}\big)$ pair --- inventory-skewed
Avellaneda--Stoikov quotes~\cite{avellaneda2008}, reinforcement-learning
policies, or simple heuristic depths --- can be paired with the same
closed-form size process. The contribution of this paper is therefore
portable across the existing literature on optimal placement: it
supplies the missing volume dimension to any pre-existing rule for
$L_{t_k}^{\pm}$.

\medskip
\noindent Empirically, on five trading days of MSFT 2014 limit-order-book data
and at the operating point $(\gamma,I_{\max})=(0.05,\,10{,}000)$, the
closed-form size policy~\eqref{eq:Q_closed_form_concl} delivers
$\$119{,}707$ of out-of-sample value per trading day at
$\max_{t}|I_{t}|=9{,}988$ shares, an
order-of-magnitude improvement over both a trivial $f\equiv 1$
baseline and a hard-bound-retrained calibration. By contrast, the two
strategies that ignore the volume dimension --- a fixed-quote
$L\equiv 1$ rule and a myopic-greedy policy --- exit the admissible
inventory band within the first trading hour and realize daily values
in the order of negative tens of millions of dollars. The closed-form
$Q^{\pm}$ is, in this sense, the unique strategy in the comparison set
that achieves the trading-revenue scale of an aggressive policy
together with the inventory-control discipline required by the
bound --- a combination that the omission of the volume dimension
cannot reproduce. Extensions to posterior uncertainty quantification
on the sdHawkes intensity, to an event space containing order
cancellations, and to joint calibration of $P_N^{\pm}$ with
$f_N^{\pm}$ are natural next steps within the present framework.

\newpage



\clearpage \pagenumbering{gobble}
\bibliographystyle{apalike}
\bibliography{bibliography}

@article{chavez2024adaptive,
  title={Adaptive optimal market making strategies with inventory liquidation cost},
  author={Ch{\'a}vez-Casillas, Jonathan and Figueroa-L{\'o}pez, Jos{\'e} E and Yu, Chuyi and Zhang, Yi},
  journal={SIAM Journal on Financial Mathematics},
  volume={15},
  number={3},
  pages={653--699},
  year={2024},
  publisher={SIAM}
}

@article{morariu2022state,
  title     = {State-Dependent Hawkes Processes and Their Application to Limit Order Book Modelling},
  author    = {Morariu-Patrichi, Maxime and Pakkanen, Mikko S.},
  journal   = {Quantitative Finance},
  volume    = {22},
  number    = {3},
  pages     = {563--583},
  year      = {2022},
  month     = {sep},
  doi       = {10.1080/14697688.2021.1977772},
  publisher = {Taylor \& Francis},
  note      = {First published online: 17 September 2021}
}

@article{capponi2021market,
  title={Market making with stochastic liquidity demand: Simultaneous order arrival and price change forecasts},
  author={Capponi, Agostino and Figueroa-L{\'o}pez, Jos{\'e} E and Yu, Chuyi},
  journal={arXiv preprint arXiv:2101.03086},
  year={2021}
}

@article{bacry2015,
  title     = {Hawkes Processes in Finance},
  author    = {Bacry, E. and Mastromatteo, I. and Muzy, J.-F.},
  journal   = {Market Microstructure and Liquidity},
  volume    = {1},
  number    = {1},
  pages     = {1--45},
  year      = {2015}
}

@article{bowser2007,
  title     = {Modelling Security Market Events in Continuous Time: Intensity-Based Multivariate Point Process Models},
  author    = {Bowsher, C. G.},
  journal   = {Journal of Econometrics},
  volume    = {141},
  number    = {2},
  pages     = {876--912},
  year      = {2007}
}

@article{cartea2015,
  title     = {Risk Metrics and Fine Tuning of High-Frequency Trading Strategies},
  author    = {Cartea, {\'A}. and Jaimungal, S.},
  journal   = {Mathematical Finance},
  volume    = {25},
  number    = {3},
  pages     = {576--611},
  year      = {2015}
}

@article{filimonov2012,
  title     = {Quantifying Reflexivity in Financial Markets: Toward a Prediction of Flash Crashes},
  author    = {Filimonov, V. and Sornette, D.},
  journal   = {Physical Review E},
  volume    = {85},
  number    = {5},
  pages     = {056108},
  year      = {2012}
}

@article{hawkes1971,
  title     = {Spectra of Some Self-Exciting and Mutually Exciting Point Processes},
  author    = {Hawkes, A. G.},
  journal   = {Biometrika},
  volume    = {58},
  number    = {1},
  pages     = {83--90},
  year      = {1971}
}

@article{merton1976,
  title     = {Option Pricing When Underlying Stock Returns Are Discontinuous},
  author    = {Merton, R. C.},
  journal   = {Journal of Financial Economics},
  volume    = {3},
  number    = {1-2},
  pages     = {125--144},
  year      = {1976}
}

@article{ogata1981,
  title     = {On Lewis' Simulation Method for Point Processes},
  author    = {Ogata, Y.},
  journal   = {IEEE Transactions on Information Theory},
  volume    = {27},
  number    = {1},
  pages     = {23--31},
  year      = {1981}
}

@article{ozaki1979,
  title     = {Maximum Likelihood Estimation of Hawkes' Self-Exciting Point Processes},
  author    = {Ozaki, T.},
  journal   = {Annals of the Institute of Statistical Mathematics},
  volume    = {31},
  number    = {1},
  pages     = {145--155},
  year      = {1979}
}

@article{confortola2014backward,
  title     = {Backward Stochastic Differential Equations and Optimal Control of Marked Point Processes},
  author    = {Confortola, Fulvia and Fuhrman, Marco},
  journal   = {SIAM Journal on Control and Optimization},
  volume    = {52},
  number    = {2},
  pages     = {858--889},
  year      = {2014},
  doi       = {10.1137/120896345},
  publisher = {SIAM},
  note      = {Received: December 17, 2012; Accepted: June 4, 2013; Published: September 24, 2013}
}

@article{confortola2016backward,
  title     = {Backward Stochastic Differential Equation Driven by a Marked Point Process: An Elementary Approach with an Application to Optimal Control},
  author    = {Confortola, Fulvia and Fuhrman, Marco and Jacod, Jean},
  journal   = {The Annals of Applied Probability},
  volume    = {26},
  number    = {3},
  pages     = {1743--1773},
  year      = {2016},
  doi       = {10.1214/15-AAP1130},
  publisher = {Institute of Mathematical Statistics},
  note      = {Received: July 2014; Revised: July 2015}
}

@article{hostoll1981,
  title     = {Optimal Dealer Pricing Under Transactions and Return Uncertainty},
  author    = {Ho, Thomas S. Y. and Stoll, Hans R.},
  journal   = {Journal of Financial Economics},
  volume    = {9},
  number    = {1},
  pages     = {47--73},
  year      = {1981},
  publisher = {Elsevier}
}

@article{avellaneda2008,
  title     = {High-Frequency Trading in a Limit Order Book},
  author    = {Avellaneda, Marco and Stoikov, Sasha},
  journal   = {Quantitative Finance},
  volume    = {8},
  number    = {3},
  pages     = {217--224},
  year      = {2008},
  publisher = {Taylor \& Francis}
}

@article{gueant2013,
  title     = {Dealing with the Inventory Risk: A Solution to the Market Making Problem},
  author    = {Gu{\'e}ant, Olivier and Lehalle, Charles-Albert and Fernandez-Tapia, Joaquin},
  journal   = {Mathematics and Financial Economics},
  volume    = {7},
  number    = {4},
  pages     = {477--507},
  year      = {2013},
  publisher = {Springer}
}

@book{cartea2015book,
  title     = {Algorithmic and High-Frequency Trading},
  author    = {Cartea, {\'A}lvaro and Jaimungal, Sebastian and Penalva, Jos{\'e}},
  publisher = {Cambridge University Press},
  year      = {2015}
}

@article{adrian2020,
  title     = {Intraday Market Making with Overnight Inventory Costs},
  author    = {Adrian, Tobias and Iyer, Tanvir and Shin, Hyun Song},
  journal   = {Journal of Financial Economics},
  volume    = {138},
  number    = {1},
  pages     = {1--21},
  year      = {2020},
  publisher = {Elsevier}
}

@article{hendershott2014,
  title     = {Price Pressures},
  author    = {Hendershott, Terrence and Menkveld, Albert J.},
  journal   = {Journal of Financial Economics},
  volume    = {114},
  number    = {3},
  pages     = {405--423},
  year      = {2014},
  publisher = {Elsevier}
}

@book{gueant2016,
  title     = {The Financial Mathematics of Market Liquidity: From Optimal Execution to Market Making},
  author    = {Gu{\'e}ant, Olivier},
  publisher = {CRC Press / Chapman and Hall},
  year      = {2016}
}

@article{bacry2016first,
  title     = {First- and Second-Order Statistics Characterization of {H}awkes Processes and Non-Parametric Estimation},
  author    = {Bacry, Emmanuel and Jaisson, Thibault and Muzy, Jean-Fran{\c{c}}ois},
  journal   = {IEEE Transactions on Information Theory},
  volume    = {62},
  number    = {4},
  pages     = {2184--2202},
  year      = {2016},
  publisher = {IEEE}
}

@article{dassios2011,
  title   = {A Dynamic Contagion Process},
  author  = {Dassios, Angelos and Zhao, Hongbiao},
  journal = {Advances in Applied Probability},
  volume  = {43},
  number  = {3},
  pages   = {814--846},
  year    = {2011},
  publisher = {Cambridge University Press}
}

@article{fushimi2024,
  title   = {Optimal High-Frequency Market Making},
  author  = {Fushimi, Takahiro and Gonz{\'a}lez Rojas, Christian and Herman, Molly},
  journal = {arXiv preprint},
  year    = {2024},
  note    = {Working paper}
}

@book{rudin1991functional,
  title     = {Functional Analysis},
  author    = {Rudin, Walter},
  edition   = {2nd},
  publisher = {McGraw-Hill},
  series    = {International Series in Pure and Applied Mathematics},
  year      = {1991}
}

@inproceedings{hansen1996adapting,
  title     = {Adapting Arbitrary Normal Mutation Distributions in Evolution Strategies: The Covariance Matrix Adaptation},
  author    = {Hansen, Nikolaus and Ostermeier, Andreas},
  booktitle = {Proceedings of IEEE International Conference on Evolutionary Computation},
  pages     = {312--317},
  year      = {1996},
  publisher = {IEEE}
}

@article{hansen2001completely,
  title     = {Completely Derandomized Self-Adaptation in Evolution Strategies},
  author    = {Hansen, Nikolaus and Ostermeier, Andreas},
  journal   = {Evolutionary Computation},
  volume    = {9},
  number    = {2},
  pages     = {159--195},
  year      = {2001},
  publisher = {MIT Press}
}

@misc{hansen2016,
  title         = {The {CMA} Evolution Strategy: A Tutorial},
  author        = {Hansen, Nikolaus},
  year          = {2016},
  eprint        = {1604.00772},
  archivePrefix = {arXiv},
  primaryClass  = {cs.LG}
}

@article{ollivier2017informationgeometric,
  title   = {Information-Geometric Optimization Algorithms: A Unifying Picture via Invariance Principles},
  author  = {Ollivier, Yann and Arnold, Ludovic and Auger, Anne and Hansen, Nikolaus},
  journal = {Journal of Machine Learning Research},
  volume  = {18},
  number  = {18},
  pages   = {1--65},
  year    = {2017}
}

@book{daley2003introduction,
  title     = {An Introduction to the Theory of Point Processes, Volume I: Elementary Theory and Methods},
  author    = {Daley, Daryl J. and Vere-Jones, David},
  edition   = {2nd},
  publisher = {Springer},
  series    = {Probability and Its Applications},
  year      = {2003}
}

@book{laplace1814essai,
  title     = {Essai philosophique sur les probabilit{\'e}s},
  author    = {Laplace, Pierre-Simon},
  year      = {1814},
  publisher = {Courcier, Paris}
}

@article{jeffreys1946invariant,
  title   = {An Invariant Form for the Prior Probability in Estimation Problems},
  author  = {Jeffreys, Harold},
  journal = {Proceedings of the Royal Society of London A},
  volume  = {186},
  number  = {1007},
  pages   = {453--461},
  year    = {1946}
}

@article{kingma2014adam,
  title   = {Adam: {A} Method for Stochastic Optimization},
  author  = {Kingma, Diederik P. and Ba, Jimmy},
  journal = {arXiv preprint arXiv:1412.6980},
  year    = {2014},
  note    = {Published at ICLR 2015}
}

@article{ogata1988statistical,
  title   = {Statistical Models for Earthquake Occurrences and Residual Analysis for Point Processes},
  author  = {Ogata, Yosihiko},
  journal = {Journal of the American Statistical Association},
  volume  = {83},
  number  = {401},
  pages   = {9--27},
  year    = {1988}
}

@book{boyd2004convex,
  title     = {Convex Optimization},
  author    = {Boyd, Stephen and Vandenberghe, Lieven},
  publisher = {Cambridge University Press},
  year      = {2004}
}

@article{almgren2001,
  author  = {Almgren, Robert and Chriss, Neil},
  title   = {Optimal Execution of Portfolio Transactions},
  journal = {Journal of Risk},
  volume  = {3},
  number  = {2},
  pages   = {5--40},
  year    = {2001}
}

@article{bowsher2007,
  author  = {Bowsher, Clive G.},
  title   = {Modelling Security Market Events in Continuous Time: Intensity Based, Multivariate Point Process Models},
  journal = {Journal of Econometrics},
  volume  = {141},
  number  = {2},
  pages   = {876--912},
  year    = {2007}
}

@article{embrechts2011,
  author  = {Embrechts, Paul and Liniger, Thomas and Lin, Lu},
  title   = {Multivariate {Hawkes} Processes: An Application to Financial Data},
  journal = {Journal of Applied Probability},
  volume  = {48},
  number  = {A},
  pages   = {367--378},
  year    = {2011}
}

@article{cont2010queue,
  author  = {Cont, Rama and Stoikov, Sasha and Talreja, Rishi},
  title   = {A Stochastic Model for Order Book Dynamics},
  journal = {Operations Research},
  volume  = {58},
  number  = {3},
  pages   = {549--563},
  year    = {2010}
}

@article{cont2014ofi,
  author  = {Cont, Rama and Kukanov, Arseniy and Stoikov, Sasha},
  title   = {The Price Impact of Order Book Events},
  journal = {Journal of Financial Econometrics},
  volume  = {12},
  number  = {1},
  pages   = {47--88},
  year    = {2014}
}

@article{lipton2014trade,
  author  = {Lipton, Alexander and Pesavento, Umberto and Sotiropoulos, Michael G.},
  title   = {Trade Arrival Dynamics and Quote Imbalance in a Limit Order Book},
  journal = {arXiv preprint arXiv:1312.0514},
  year    = {2014}
}

@article{stoikov2018microprice,
  author  = {Stoikov, Sasha},
  title   = {The Micro-Price: A High-Frequency Estimator of Future Prices},
  journal = {Quantitative Finance},
  volume  = {18},
  number  = {12},
  pages   = {1959--1966},
  year    = {2018}
}

@inproceedings{spooner2018marketmaking,
  author    = {Spooner, Thomas and Fearnley, John and Savani, Rahul and Koukorinis, Andreas},
  title     = {Market Making via Reinforcement Learning},
  booktitle = {Proceedings of the 17th International Conference on Autonomous Agents and Multiagent Systems (AAMAS)},
  pages     = {434--442},
  year      = {2018}
}

@article{gueant2019rl,
  author  = {Gu\'{e}ant, Olivier and Manziuk, Iuliia},
  title   = {Deep Reinforcement Learning for Market Making in Corporate Bonds: Beating the Curse of Dimensionality},
  journal = {Applied Mathematical Finance},
  volume  = {26},
  number  = {5},
  pages   = {387--452},
  year    = {2019}
}

@article{hambly2023rl,
  author  = {Hambly, Ben and Xu, Renyuan and Yang, Huining},
  title   = {Recent Advances in Reinforcement Learning in Finance},
  journal = {Mathematical Finance},
  volume  = {33},
  number  = {3},
  pages   = {437--503},
  year    = {2023}
}

@article{hansen2001cmaes,
  author  = {Hansen, Nikolaus and Ostermeier, Andreas},
  title   = {Completely Derandomized Self-Adaptation in Evolution Strategies},
  journal = {Evolutionary Computation},
  volume  = {9},
  number  = {2},
  pages   = {159--195},
  year    = {2001}
}

@article{hansen2016tutorial,
  author  = {Hansen, Nikolaus},
  title   = {The {CMA} Evolution Strategy: A Tutorial},
  journal = {arXiv preprint arXiv:1604.00772},
  year    = {2016}
}

@misc{cartea2021gradient,
  author        = {Cartea, {\'A}lvaro and Cohen, Samuel N. and Labyad, Saad},
  title         = {Gradient-based estimation of linear {Hawkes} processes with general kernels},
  year          = {2021},
  eprint        = {2111.10637},
  archivePrefix = {arXiv},
  primaryClass  = {stat.ME},
  doi           = {10.48550/arXiv.2111.10637},
  url           = {https://arxiv.org/abs/2111.10637},
  note          = {arXiv preprint arXiv:2111.10637}
}

@article{bremaud1996stability,
  author    = {Br{\'e}maud, Pierre and Massouli{\'e}, Laurent},
  title     = {Stability of nonlinear {Hawkes} processes},
  journal   = {The Annals of Probability},
  volume    = {24},
  number    = {3},
  pages     = {1563--1588},
  year      = {1996},
  doi       = {10.1214/aop/1065725193},
  publisher = {Institute of Mathematical Statistics}
}

\newpage

\typeout{=== BIBTEX DEBUG ===}
\nocite{*} 
\typeout{=== END DEBUG ===}
\end{document}